\documentclass[12pt]{amsart}
\allowdisplaybreaks[1]

\usepackage[usenames]{color}

\usepackage{a4wide}
\usepackage{amssymb}
\usepackage{amsmath}
\usepackage{amscd}
\usepackage{amsfonts,mathrsfs}
\usepackage{hyperref}

\theoremstyle{plain}
\newtheorem{lemma}{Lemma}[section]
\newtheorem{prop}[lemma]{Proposition}
\newtheorem{coro}[lemma]{Corollary}
\newtheorem{thm}[lemma]{Theorem}

\theoremstyle{definition}
\newtheorem{example}[lemma]{Example}
\newtheorem{definition}[lemma]{Definition}
\newtheorem{remark}[lemma]{Remark}

\newcommand{\ts}{\hspace{0.5pt}}

\newcommand{\N}{\mathbb{N}}
\newcommand{\R}{\mathbb{R}\ts}

\newcommand{\q}[1]{q_{ #1}}

\newcommand{\D}{D}
\newcommand{\F}{F}
\newcommand{\Lp}{L}
\newcommand{\e}{\mathrm{e}}

\newcommand{\al}{\alpha}
\newcommand{\de}{\delta}
\newcommand{\eps}{\varepsilon}
\newcommand{\te}{\theta}
\newcommand{\lm}{\lambda}
\newcommand{\si}{\sigma}
\newcommand{\ph}{\varphi}

\newcommand{\aV}[1]{\left\Vert #1\right\Vert}
\newcommand{\as}[1]{\langle #1\rangle}
\newcommand{\ov}[1]{\overline{#1}}
\newcommand{\ow}[1]{\widetilde{#1}}

\newcommand{\Hmm}[1]{\leavevmode{\marginpar{\tiny%
$\hbox to 0mm{\hspace*{-0.5mm}$\leftarrow$\hss}%
\vcenter{\vrule depth 0.1mm height 0.1mm width \the\marginparwidth}%
\hbox to 0mm{\hss$\rightarrow$\hspace*{-0.5mm}}$\\\relax\raggedright #1}}}

\newcommand{\IC}{\mathbb{C}}
\newcommand{\IR}{\mathbb{R}}
\newcommand{\ISS}{\mathscr{S}}
\newcommand{\IFF}{\mathscr{F}}

\newcommand{\IN}{\mathbb{N}}

\newcommand{\Id}{{\rm d}}

\newcommand{\f}{\frac}

\newcommand{\nn}{\nonumber}

\begin{document}

\title[A Feynman-Kac-It\^{o} Formula]{A Feynman-Kac-It\^{o} Formula for magnetic Schr\"odinger operators on graphs}

\author[B. G\"uneysu]{Batu G\"uneysu}
\address{Batu G\"uneysu, Institut f\"ur Mathematik, Humboldt-Universit\"at zu Berlin, 12489 Berlin, Germany} \email{gueneysu@math.hu-berlin.de}
\author[M. Keller]{Matthias Keller}
\address{Matthias Keller, Einstein Institute of Mathematics, The Hebrew University of Jerusalem,
Jerusalem 91904, Israel}\email{mkeller@ma.huji.ac.il}

\author[M. Schmidt]{Marcel Schmidt}
\address{Marcel Schmidt, Mathematisches Institut, Friedrich-Schiller-Universit\"at Jena,
07743 Jena, Germany}
\email{schmidt.marcel@uni-jena.de}

\begin{abstract}
In this paper we prove a Feynman-Kac-It\^{o} formula for magnetic Schr\"odinger operators on arbitrary weighted graphs. To do so, we have to provide a natural and general framework both on the operator theoretic and the probabilistic side of the equation.  On the operator side we identify a very general class of potentials  that allows the definition of magnetic Schr\"odinger operators. On the probabilistic side, we introduce an appropriate notion of stochastic line integrals with respect to magnetic potentials.  Apart from linking the world of discrete magnetic operators with the probabilistic world through the  Feynman-Kac-It\^{o} formula, the insights from this paper gained on both sides should be of an independent interest.
As applications of the Feynman-Kac-It\^{o} formula, we prove a Kato inequality, a Golden-Thompson inequality and an explicit representation of the quadratic form domains corresponding to a large class of potentials.
\end{abstract}
\date{\today} %
\maketitle

\section{Introduction}

The conceptual importance of the classical Feynman-Kac formula stems from the fact that it links the world of operator theory (or partial differential equations) with that of probability. In particular, the semigroup of a Schr\"odinger operator of the form $-\Delta+v$ on $L^2(\IR^n)$ is expressed in terms of an expectation value involving the Markov process of the free operator $-\Delta$,  which is nothing but the Euclidean Brownian motion in this case. If one perturbs $-\Delta+v$ by a magnetic field with potential $\theta$, one has to deal with the magnetic Schr\"odinger operator  $-\Delta_{\theta}+v$. In this case a very important extension of the Feynman-Kac formula is given by the Feynman-Kac-It\^{o} formula. This formula again expresses the semigroup corresponding to the latter operator through Euclidean Brownian motion, where now one has to take the (Stratonovic) stochastic line integral of $\theta$ along the Brownian motion into account \cite{Si05}. Such probabilistic representations have many important physical consequences through diamagnetism, e.g., one can easily deduce that switching on a magnetic field can only lead to an increase of the ground state energy of the systems.\medskip

Seeking for extensions of the above results to more general settings than the Euclidean $\IR^n$, one will realize that the \emph{Feynman-Kac} formula can be proven for locally compact regular Dirichlet spaces (see, e.g., \cite{demuth}), where one simply has to replace $-\Delta$ with the operator corresponding to the given Dirichlet form, and Brownian motion with the associated Markov process. However, it is not even clear how to formulate a Feynman-Kac-\emph{It\^{o}} formula in many situations. The reason for this is that a consistent theory of Schr\"odinger operators with local magnetic potentials in such a general setting as Dirichlet spaces is still missing. Although recently very promising progress into this direction has been made on the operator side \cite{CS,HRT,HT}, there still remains the issue of finding a reasonable way to define a proper notion of a stochastic line integral which extends the  $\IR^n$-theory in a consistent way.\medskip

The situation is fundamentally better for smooth Riemannian manifolds $M$. Here, magnetic potentials can be defined simply as real-valued $1$-forms. If $\theta$ is such a $1$-form, then $-\Delta_{\theta}+v$ can be defined invariantly in analogy to the Euclidean case (see for example \cite{bg,shub} for details). Assuming some local control on $v_-$ (typically $L^{1}_{loc}$) and $\theta$ (typically smooth or $L^{p}_{loc}$), and a certain global control on $v_-$, the operator $-\Delta_{\theta}+v$ will correspond to a well-defined self-adjoint semi-bounded operator on $L^2(M)$. One can then prove an analogue of the Feynman-Kac-It\^{o} formula in this setting (replacing the Euclidean with the underlying Riemannian Brownian motion) without any further assumptions on $M$. As the underlying manifold locally looks like the \emph{linear} space $\R^{n}$, one can define the line integral of $\theta$ along the Riemannian Brownian motion by combining the definition from the Euclidean case either with a patching procedure using charts \cite{IW}, or equivalently, by embedding $M$ into some $\R^{l}$ with an appropriate $l\ge n$, as in \cite{Em}. As a consequence of the Feynman-Kac-It\^o formula in this setting it becomes very easy to deduce several rigorous variants
of the domination \lq\lq{}$-\Delta_{\theta}+v\geq -\Delta+v$\lq\lq{}. Apart from physically relevant ones, these domination results also make it possible to transfer many important mathematical statements from zero magnetic potential to arbitrary magnetic potentials, such as essential self-adjointness results \cite{bg,Si85} or certain smoothing properties of the Schr\"odinger semigroups \cite{Bro,bg}.\medskip

Going back to the fundamental papers \cite{Ha,LL,Su}, there is also a basic theory of magnetic Schr\"odinger operators for discrete graphs. In the last years an extensive amount of research for these operators has been carried out into various directions. Let us only mention here that basic spectral properties and Kato's inequality have been proven in \cite{DM}, for a Hardy inequality see \cite{Gol}, for approximation results of spectral invariants see \cite{MSY, MY}, and for weak Bloch theory see \cite{HS}. Recently there has been a strong focus on the question of essential self-adjointness of magnetic Schr\"odinger operators  \cite{CdVT-HT,Gol,Mi,Mi2,MT,Sus}. \medskip

On discrete graphs the Markov processes corresponding to free Laplacians  are jump processes (which have very special path properties), and  magnetic potentials are typically defined as functions on the underlying set of edges. So, one might hope that it is possible to get a proper notion of line integrals in this setting, which produces a probabilistic representation of the magnetic Schr\"odinger semigroups. The main result of this paper, a Feynman-Kac-It\^{o} type formula for discrete graphs, precisely states that this is possible.\medskip

Unfortunately, so far all proposed settings for discrete magnetic Schr\"odinger operators are somewhat tailored to their specific applications and, thus, are often rather restrictive.  In particular, a general and systematic treatment of the question, when the operators can actually be defined as genuine self-adjoint operators, seems to be missing.

\medskip

The first question that arises is actually what a natural and sufficiently general framework might be in this context. We start with quadratic forms associated with graphs and then identify a class of potentials that is suitable to our cause. Having the goal of a Feynman-Kac-\emph{It\^o} formula in mind (where due to the presence of a magnetic potential one cannot expect to conclude exclusively with monotone convergence arguments), a natural assumption on the potential is that the corresponding non-magnetic quadratic form is semi-bounded from below on the functions with compact support. Remarkably, it turns out that the latter assumption is in fact all we need to get a closable semi-bounded form, and thus a self-adjoint semi-bounded operator, in the magnetic case. This is the content of Theorem~\ref{t:closable}. To the best of our knowledge, this result is even new in the non-magnetic setting { since it goes beyond classical perturbation theory in the spirit of Kato.

\medskip

In addition, we} give criteria for the above mentioned self-adjoint semi-bounded operator to be unique in an appropriate sense, which in turn also provides criteria  for a certain uniqueness of the Markov processes.\medskip

Having established the operator theoretic side, we then give the definition of the stochastic line integral in terms of a sum along the path of the process. We establish our main result, the Feynman-Kac-It\^o formula in Theorem~\ref{main}. Let us stress that we do not have to make any restrictions on the underlying geometry such as local finiteness of the graph. Furthermore, we do not require anything on the positive parts of the potentials, nor on the magnetic potentials. The only assumption we make on the negative part of the potential is the already mentioned, namely, that the corresponding non-magnetic form is semi-bounded from below on the functions with compact support.  Compared to the manifold case this  assumption is significantly weaker, obviously due to the discrete structure of our setting.\medskip

Finally, we remark that manifolds and graphs essentially provide the most approachable and prominent non-trivial examples of local and non-local Dirichlet forms. So, having established a Feynman-Kac-It\^o formula in both of these worlds appears to be a promising step towards a unified theory for all regular Dirichlet forms. Here, as we have already mentioned, the results of \cite{CS,HRT,HT} should be very useful.\medskip

The paper is structured as follows: In Section~\ref{sett}, we introduce and establish all necessary operator theoretic results. In Section~\ref{proc}, we introduce the necessary probabilistic concepts (including the definition of the line integral in this setting). Section~\ref{haup} is completely devoted to the presentation and the proof of our main result, the Feynman-Kac-It\^{o} formula, Theorem~\ref{main}, and finally, in Section~\ref{anwen}, we have collected several applications such as semigroup formulas, Kato's inequality, a Golden-Thompson inequality and a representation of the form domain for suitable potentials.\vspace{2mm}

\emph{Note added:} Let us mention the follow up papers by the first named author \cite{semic,GM} which treat Feynman-Kac formulae and semiclassical limits for covariant Schr\"odinger semigroups on Hermitian vector bundles over infinite weighted graphs. These papers are heavily building on the results presented here.

\section{Magnetic Schr\"odinger operators}\label{sett}

In this section we introduce the set up in which we are going to prove the  Feynman-Kac-It\^{o} formula. While it is clear from earlier work how a magnetic Schr\"odinger operator should act \cite{Ha,LL,Mi,Su}, it is a non-trivial problem to determine when a self-adjoint semi-bounded operator can be defined. This starts with the problem that for general weighted graphs the formal operator does not necessarily map the compactly supported functions into $\ell^{2}$.
Although the theory of quadratic forms provides a helpful tool, it raises the problem of determining whether the form, defined a priori on the compactly supported functions, is closable and semi-bounded from below. This, however, is a rather subtle issue which in general does not allow for a complete and applicable characterization. Here, we provide a rather general framework in which we give a sufficient condition (cf. Theorem~\ref{t:closable} below) for the general magnetic case, which, remarkably, even turns out to be necessary in the non-magnetic case. Interestingly, we will use a Feynman-Kac-It\^o formula for potentials that are bounded below in order to derive the latter result.

After briefly reviewing the basic set up of weighted graphs, we introduce the matgnetic forms and the corresponding formal operators. Then, we give a sufficient criterion for closability and semi-boundedness of the forms. At the end, we discuss uniqueness of semi-bounded self-adjoint extensions/restrictions and present a result on semigroup convergence.

\subsection{Weighted graphs}\label{s:graphs} We essentially follow the setting of \cite{KL}. Let $(X,b)$ be a graph, that is, $X$ is a countable  set, {equipped with the discrete topology,} and
\[
b:X\times X \longrightarrow  [0,\infty)
\]
is a symmetric function with the properties $b(x,x)=0$ and
\begin{align*}
\>\sum_{y \in X}b(x,y) < \infty \>\text { for all }x \in X.
\end{align*}

Then, the elements of $X$ are called \emph{vertices} and one says that $x,y \in X$ are \emph{neighbors} or \emph{connected by an edge}, if $b(x,y) >0$, which is written as $x \sim y$.
The graph $X$ is called \emph{locally finite}, if every vertex has only a finite number of neighbors. Furthermore, a \emph{path} on the graph $X$ is a (finite or infinite) sequence of pairwise distinct vertices $(x_{j})$ such that $x_j \sim x_{j+1}$ for all $j$, and $X$ is called \emph{connected}, if for any $x,y\in X$ there is a path $(x_j)^n_{j=0}$ such that $x_0=x$ and $x_n=y$.
\medskip

{For simplicity and without loss of generality,
we will \emph{assume} throughout the paper that the graph $X$ is \emph{connected}.}
\medskip

When $X$  is equipped with the discrete topology, any function $m: X \to (0,\infty)$ gives rise to a Radon measure of full support on $X$ by setting $m(A) := \sum_{x \in A}m(x)$. Then, the triple $(X,b,m)$ is called a \emph{weighted graph}.
For $x\in X$, we denote the \emph{weighted vertex degree} by
$$\text{deg}_m(x)= \frac{1}{m(x)}\sum_{y \in X} b(x,y).$$
{Often, we  use this notation for the constant measure $m\equiv1$ in which case we have $\deg_{1}(x)=\sum_{y\in X}b(x,y)$.}
This notion is motivated by the following observation: Whenever $m\equiv 1$ and $b:X\times X\to \{0,1\}$, the number $\text{deg}_m(x)=\text{deg}_1(x)$ is equal to the number of edges emerging from a vertex $x$.

\subsection{Quadratic forms}

Let ${C}(X)$ be the linear space of all complex-valued  functions on $X$ and ${C}_c(X)$ its subspace of functions with finite support. We denote the standard scalar product and norm on $\ell^2(X,m)$ with $\as{\bullet,\bullet}$ and $\aV{\bullet}$, respectively, that is,
\begin{align*}
\as{f,g}=\sum_{x\in X}f(x)\overline{g(x)}m(x), \qquad\aV{f}=\as{f,f}^{\frac{1}{2}}.
\end{align*}
Clearly, ${C}_c(X)$ is dense in $\ell^2(X,m)$. Let $\de_{x}$ be the function that takes the value $1/m(x)$ at $x$ and $0$ otherwise. By the discreteness of the underlying data, any linear operator $A$ in $\ell^2(X,m)$ with $C_c(X)\subseteq D(A)$ has a unique integral kernel in the sense that the function
\[
A(\bullet,\bullet):X\times X\longrightarrow \IC, \qquad A(x,y)=\f{1}{m(x)}\left\langle A\delta_x,\delta_y\right\rangle
\]
is the unique one such that
\[
Af(x)=\sum_{y\in X}A(y,x)f(y) m(y)\>\>\text{ for all $f\in\D(A)$, $x\in X$.}
\]

Following \cite{CdVT-HT}, we understand by a \emph{magnetic potential} on the set $X$ a function $$\theta: X \times X \to [-\pi,\pi]\quad\mbox{such that}\quad\theta(x,y) = - \theta(y,x),\;x,y\in X.$$ A function $v: X \to \R$ will be simply called a \emph{potential}.
\medskip

{Throughout the paper, let $\theta$ be an\emph{ arbitrary magnetic potential},  and \emph{if not further specified}, then $v$ denotes an \emph{arbitrary potential}.}
\medskip

We define a symmetric sesqui-linear form on $\ell^2(X,m)$ with domain of definition ${C}_c(X)$ by
\begin{align*}
Q_{v,\theta}^{(c)}(f,g):= & \frac{1}{2}\sum_{x,y \in X} b(x,y)\big(f(x)-\e^{\mathrm{i} \theta(x,y) }f(y)\big)\overline{\big(g(x)-\e^{\mathrm{i} \theta(x,y)} g(y)\big)}+\sum_{x\in X}v(x)f(x)\overline{g(x)}m(x)\nn.
\end{align*}
With
$$\ow{{\F}}(X):= \left\{f\in {C}(X)\left|\sum_{y\in X} b(x,y)|f(y)| < \infty \text{ for all }x\in X\right\}\right.,$$
we define the formal difference operator $\ow{L}_{v,\theta}:\ow{{\F}}(X) \to {C}(X)$ by
$$\ow{L}_{v,\theta}f(x) = \frac{1}{m(x)}\sum_{y \in X} b(x,y)\big (f(x)-\e^{\mathrm{i} \theta(x,y) }f(y)\big)+v(x)f(x).$$
Note that if $X$ is locally finite, then one has $\ow{{\F}}(X)={C}(X)$. However, in general, $\ow{{\F}}(X)$ does not  include $\ell^{2}(X,m)$.

The form $Q_{v,\theta}^{(c)}$ and the operator $\ow L_{v,\theta}$ are related by  Green's formula. {We give two formulations: One for a very large class of functions which does not allow for an expression in terms of the introduced forms and scalar products but only explicitly in terms of sums. The second is a concise formulation for compactly supported functions.}

\begin{lemma} \label{Greens formula}\emph{(Green's formula)} For all $f\in \ow{{\F}}(X)$, $g \in {C}_c(X)$, one has
\begin{align*}
\sum_{x\in X}&\ow{L}_{v,\theta}f(x)\overline{g(x)}m(x) = \sum_{x\in X} f(x)\overline{\ow{L}_{v,\theta}g(x)}m(x)\\
&=\frac{1}{2}\sum_{x,y \in X} b(x,y) \Big(f(x)-\e^{\mathrm{i} \theta(x,y) }f(y)\Big)\overline{\Big(g(x)-\e^{\mathrm{i} \theta(x,y) }g(y)\Big)} + \sum_{x\in X}v(x) f(x)\overline{g(x)} m(x).\nn
\end{align*}
Moreover, if $\ow L_{v,\theta}[{C}_c(X)]\subseteq \ell^{2}(X,m)$, then for all $f,g\in C_{c}(X)$ one has
\begin{align*}
    Q^{(c)}_{v,\theta}(f,g)=\langle \ow L_{v,\theta} f, g\rangle=\langle f,\ow L_{v,\theta} g\rangle.
\end{align*}
\end{lemma}
\begin{proof} The statements follow from a direct computation, where absolute convergence of the sums is guaranteed by $g\in {C}_c(X)$ and $\sum_{y}b(x,y)|f(y)|<\infty$ as $f\in \ow{{\F}}(X)$  (cf. \cite[Lemma~4.7]{HK} for more details).
\end{proof}

If $Q^{(c)}_{v,\theta}$ is semi-bounded from below and closable, we denote its  closure by $Q_{v,\theta}$ with domain $D(Q_{v,\theta})$ and the corresponding self-adjoint operator by $L_{v,\theta}$ with domain $D(L_{v,\theta})$, see  \cite[Theorem~VIII.15]{RS}.

{ When  $Q_{v,\theta}^{(c)}$ is semi-bounded from below and closable, it is known in many cases that the domain of $Q_{v,\theta}$ is contained in $\ow{F}(X)$. In this case an important consequence of Green\rq{}s formula is that the corresponding self-adjoint operator $L_{v,\theta}$ is in fact a restriction of $\ow L_{v,\theta}$, see Theorem~\ref{operator} below.}

{For $v\ge0$ and $\theta\equiv0$, the form $Q_{v,0}^{(c)}$ is always closable on $\ell^{2}(X,m)$ and its closure is a regular Dirichlet form as $X$ is equipped with the discrete topology. Indeed, all regular Dirichlet forms on discrete measure spaces $(X,m)$ are parameterized by graphs $b$ and potentials $v\ge0$. This situation was studied in \cite{KL} to which we refer the interested reader  for details. In what follows, we use the conventions $Q:=Q_{0,0}$ and $L:=L_{0,0}$.}

{\begin{remark} Note that suitable extensions of $Q$ to the space of functions of finite energy $$\{f\in C(X) \mid \sum_{x,y\in X}b(x,y)|f(x)-f(y)|^{2}<\infty\}$$ are resistance forms in the sense of Kigami \cite{Kig01}. Indeed, the only assumption in \cite[Definition~2.3.1]{Kig01} which is non-trivial to check is (RF04). This however follows by \cite[Lemma~3.4]{GHKLW}. On the other hand, a magnetic form $Q^{(c)}_{0,\theta}$, $\theta\neq0$, can not be extended to a resistance form since it violates the cut-off property (RF05).
\end{remark}}


\subsection{Potentials}
{We are interested in classes of potentials $v$ such that  the forms $Q_{v,\theta}^{(c)}$ are semi-bounded from below and closable.  For this some restrictions on the potentials are needed.

In the sequel, whenever dealing with a sesqui-linear form $s$, we denote  its associated quadratic form by the same letter, i.e., $s(f):=s(f,f)$ for $f$ in the domain of $s$. Moreover, from now on the term ''semi-bounded`` will always mean ''semi-bounded from below``.}

For a function $w:X\to\IR$, we will write $w_{\pm}=(\pm w)\vee0$ such that $w=w_{+}-w_{-}$.

Let $\q{v}$ be the symmetric sesqui-linear form given by $v$, that is,
\begin{align*}
\D(\q{v}):=\ell^{2}(X,|v|m){\cap \ell^{2}(X,m)},\quad\q{v}(f,g)&:=\sum_{x\in X}v(x)f(x)\overline{g(x)}m(x).
\end{align*}
We consider the following  classes of potentials
\begin{align*}
    \mathcal{A}_{\theta}:=\left.\Big\{w:X\to\R \right| \>
    &\mbox{There is $C\geq 0$ such that } \\
     &\q{w_{-}}(f)\leq Q_{w_{+},\theta}^{(c)}(f)+C\|f\|^2 \mbox{ for all $f\in C_c(X)$}\Big\}
\end{align*}
and
\begin{align*}
   \mathcal{B}_{\theta}:=\left.\Big\{w:X\to\R \right| \>
&\mbox{There are $\eps>0$ and $C\geq 0$ such that}\\
    & \q{w_{-}}(f)\leq (1-\eps) Q_{w_{+},\theta}^{(c)}(f)+C\|f\|^2\mbox{ for all $f\in C_c(X)$}\Big\}.
\end{align*}
{
First, we  observe the obvious inclusions
\begin{align*}
    \mathcal{B}_{\theta}\subseteq \mathcal{A}_{\theta}.
\end{align*}
The  potentials $v$ in $\mathcal{B}_{\theta}$ give rise to forms $q_{v_{-}}$ which are called \emph{infinitesimally form bounded} with respect to $Q_{v_{+},\theta}^{(c)}$ in the literature. For this class one can apply classical perturbation theory in the spirit of Kato, see  Proposition~\ref{p:closable}.

The importance of the larger class $\mathcal{A}_{\theta}$ stems from the following elementary observation.

\begin{lemma}\label{l:potentials}
The form $Q^{(c)}_{v,\theta}$ being semi-bounded  is equivalent to $v\in\mathcal{A}_{\theta}$. Moreover,
 \begin{align*}
    \mathcal{A}_{0}\subseteq \mathcal{A}_{\theta}\quad\mbox{and}\quad \mathcal{B}_{0}\subseteq \mathcal{B}_{\theta}.
 \end{align*}
 In particular,  for any $v\in \mathcal{A}_{0}$ the form $Q^{(c)}_{v,\theta}$ is semi-bounded.
\end{lemma}
\begin{proof} The first statement is straightforward from the definition. For the second statement note that for $f \in C_c(X)$, we obviously have
$ q_{v_-}(f)  =  q_{v_-}(|f|) $  and  $ Q_{{v_+},0}^{(c)}(|f|) \leq Q_{{v_+},\theta}^{(c)}(f)$.
The ''in particular`` part is clear.
\end{proof}

In the remark below we give a preview of what we will prove for the potentials in each of these classes. Let us note that these considerations go beyond standard perturbation theory.

\begin{remark}In the sequel we will encounter the following configurations of assumptions:
\begin{itemize}
  \item [(A)] $v\in \mathcal{A}_{0}$.
  \item [(B)] $v\in \mathcal{A}_{\theta}$ and $b$ locally finite.
  \item [(C)] $v\in \mathcal{B_{\theta}}$.
\end{itemize}
In many cases in (B) it even suffices to assume $\ow L_{v,\theta} [C_{c}(X)]\in\ell^{2}(X,m)$ which is implied by local finiteness, see Lemma~\ref{finiteness assumption}.

By Lemma \ref{l:potentials} we know that even for the largest class $\mathcal{A}_{\theta}$ the forms $Q_{v,\theta}^{(c)}$ are semi-bounded. So, a natural starting point is closability of the form $Q_{v,\theta}^{(c)}$ in $\ell^{2}(X,m)$. This is implied by any of the assumptions above, i.e.,
\begin{align*}
    \mbox{(A) or (B) or (C)}\;\Longrightarrow \;\mbox{closability of $Q_{v,\theta}^{(c)}$}.\quad (\mbox{Proposition~\ref{p:closable}, Theorem~\ref{t:closable}})
\end{align*}
Having a closable form, we can consider the closure $Q_{v,\theta}$ which comes with an associated positive self-adjoint operator $L_{v,\theta}$ by general theory. For various considerations it is important to know the action of the operator $L_{v,\theta}$. By Green's formula we know that $L_{v,\theta}$ is a restriction of $\ow L_{v,\theta}$ on $C_{c}(X)$ (whenever $C_{c}(X)$ is included in $D(L)$), so, it would be desirable to know whether $L_{v,\theta}$ is a restriction of $\ow L_{v,\theta}$ on $D(L)$. This is indeed guaranteed under the assumptions (B) and (C), i.e.,
\begin{align*}
    \mbox{(B) or (C)}\;\Longrightarrow \;\mbox{$L_{v,\theta}$ is a restriction of $\ow L_{v,\theta}$ on $D(L)$}.\quad (\mbox{Theorem~\ref{operator}})
\end{align*}
Furthermore, under additional assumptions we can show that $L_{v,\theta}$ is the unique self-adjoint restriction of $\ow L_{v,\theta}$ on $\ell^{2}(X,m)$. This is a slight generalization of the concept of essential self-adjointness. In Section~\ref{s:uniqueness} we proof such results under the assumptions (B) and (C)
\begin{align*}
    \mbox{(B) or (C)}\;\Longrightarrow \;\mbox{Uniqueness results for $L_{v,\theta}$}.\quad (\mbox{Theorem~\ref{t:essSA1} and Theorem~\ref{t:essSA2}})
\end{align*}
The major result of this paper is  a Feynman-Kac-It\^o formula. Here, assumption (A) suffices, as it guarantees closability of the forms $Q_{v,\theta}^{(c)}$ and $Q_{v,0}^{(c)}$,
\begin{align*}
    \mbox{(A)}\;\Longrightarrow \;\mbox{Feynman-Kac-It\^o formula for $\mathrm{e}^{-tL_{v,\theta}}$}\quad (\mbox{Theorem~\ref{main}})
\end{align*}
It is remarkable that we do not need any explicit knowledge of the operator to prove this result. In particular assumptions (B) or (C) which guarantee such knowledge do not enter.

Finally, let us mention a situation under which all results of the paper hold:
\begin{itemize}
  \item [(D)] $v\in \mathcal{B}_{0}$.
\end{itemize}
\end{remark}

}

{
Let us discuss some examples for these classes of potentials. We start by an
important subclass of  $\mathcal{B}_{0}$ -- the Kato class. Then, we give examples on the threshold of $\mathcal{B}_{0}$ and $\mathcal{A}_{0}$. Finally, we refer to the general framework of admissible potentials, where $\mathcal{A}_{0}$ can already seen to be featured.

\begin{example}[Kato class]\label{kkato} We recall that a function $w:X\to \IR$ is in the \emph{Kato class} $\mathcal{K}$ of the regular Dirichlet form $Q=Q_{0,0}$, if and only if
\begin{align*}
\lim_{t\to 0+}\sup_{x\in X}\int^t_0\sum_{y\in X} \mathrm{e}^{-s L}(x,y) |w(y)| m(y)\Id s =0,
\end{align*}
where $(x,y)\mapsto e^{-tL}(x,y)$ is the kernel of the semigroup of the operator $L=L_{0,0}$ associated to $Q$.
By combining \cite[Lemma~3.1]{KT} with \cite[Theorem~3.1]{peter}, we immediately get $\mathcal{K}\subseteq \mathcal{B}_{0}$. So, one has
$$
\ell^{p}(X)\subseteq \mathcal{K}\subseteq \mathcal{B}_{0}\text{ for all $p\in [1,\infty]$},
$$
where of course $\ell^p(X)=\ell^p(X,m\equiv 1)\subseteq \ell^{\infty}(X)$. This follows from the uniform estimates
$$
\mathrm{e}^{-sL}(x,y)m(y)\leq\sum_{z\in X}\mathrm{e}^{-sL}(x,z)m(z)\leq 1,\quad\text{ $s\geq 0$, $x,y\in X$.}
$$
Here, the second inequality follows as $Q$ is a Dirichlet form. We would like to stress the fact that the validity of the inclusion $\ell^p(X)\subseteq \mathcal{K}$ \emph{without any further assumptions on $Q$} is a  special feature of discrete spaces, in the sense that on Riemannian manifolds one needs considerable curvature assumptions to produce $L^p$-type subspaces of the corresponding Kato class for $p\ne \infty$ (cf. \cite{KT2}).
\end{example}
}
{
Let us come to examples at the threshold of $\mathcal{A}_{0}$ and $\mathcal{B}_{0}$.
\begin{example}\label{ex:deg}
Let $(X,b)$ be a graph and $m$ a measure. Recall that Cheeger's constant is given by
\begin{align*}
    \al :=\inf_{W\subseteq X\mbox{\scriptsize{finite}}}\frac{b(\partial W)}{\deg_{1}(W)},
\end{align*}
where $\partial W= W\times (X\setminus W)$ and $\deg_{1}(x)=\sum_{y\in X}b(x,y)$. Then, for $\ph\in C_{c}(X)$ one has the following inequality, \cite{Gol,KL2},
\begin{align*}
    (1-\sqrt{1-\al^{2}})q_{\deg_{m}}(\ph)\leq Q(\ph).
\end{align*}
In the case $\al>0$, consider
\begin{align*}
    v_{\eps} :=-\frac{(1-\eps)}{(1-\sqrt{1-\al^{2}})}\deg_{m}, \quad \eps\ge0.
\end{align*}
Then, $v_{\eps}\in \mathcal{B}_{0}\subseteq \mathcal{A}_{0}$
for $\eps>0$ and, for $\eps=0$, we have
$v_{0}\in \mathcal{A}_{0}$. \\
In order to present an example in $\mathcal{A}_{0}\setminus \mathcal{B}_{0}$ consider a binary tree with standard weights, i.e., $b(x,y)\in\{0,1\}$. In this case, the inequality above is sharp independent of the choice of the measure $m$. Moreover, for a probability measure $m$ the form $q_{v_{0}}$ for the potential $v_{0}$ is unbounded on $\ell^{2}(X,m)$. Thus, $v_{0}\in \mathcal{A}_{0}\setminus \mathcal{B}_{0}$.
\end{example}

Finally, we address a rather abstract class of potentials which includes $\mathcal{A}_{0}$.
\begin{example}
 As it can be seen from Theorem~\ref{t:closable} below, the class of \emph{admissible potentials} corresponding to the closure of  $Q_{v_{+},0}^{(c)}$, which has been introduced in \cite{Vo1,Vo2} includes $\mathcal{A}_{0}$.  See \cite{KLVW} for further characterizations of the class of admissible potentials.
\end{example}

In summary
\begin{align*}
\mathcal{K}\subseteq \mathcal{B}_{0}\subseteq\mathcal{A}_{0}\subseteq \{\mbox{admissible potentials}\},
\end{align*}
where there are examples such that the inclusion in the middle is strict.

By the inclusions $\mathcal{A}_{0}\subseteq \mathcal{A}_{\theta}$,  $\mathcal{B}_{0}\subseteq \mathcal{B}_{\theta}$, discussed in Lemma~\ref{l:potentials}, the potentials in Examples~\ref{kkato} and~\ref{ex:deg} are also examples for $\theta\neq0$. Nevertheless, the classes $\mathcal{A}_{\theta}$ and $\mathcal{B}_{\theta}$ may depend on the parameter $\theta$.
}


\subsection{Semi-boundedness, closability and associated operators}\label{s:forms}

In this subsection we state the result that for potentials in $\mathcal{A}_{0}$ the corresponding magnetic quadratic form is closable.

By making suitable assumptions on the geometry of $(X,b)$ or on the negative part of the potential we can determine the action of the operator  associated to the closure of $Q^{(c)}_{v,\theta}$. It turns out that in many cases  this operator is a restriction of $\ow{L}_{v,\theta}$ (see Theorem \ref{operator}).

{ We start with an observation which does not come as a surprise from the perspective of classical perturbation theory in the spirit of Kato. However, we can not simply give a reference since we use the explicit action of the form.}
{
\begin{prop}\label{p:closable} For $v \in \mathcal{B}_{\theta}$ the form $Q^{(c)}_{v,\theta}$ is semi-bounded and closable. Its closure $Q_{v,\theta}$ is semi-bounded and given by $Q_{v,\theta} =  Q_{v_+,\theta} - q_{v_-}$ with domain $D(Q_{v,\theta}) = D(Q_{v_+,\theta})$. Furthermore, for all $f,g \in D(Q_{v,\theta})$, one has
$$Q_{v,\theta}(f,g) = \frac{1}{2}\sum_{x,y \in X}b(x,y)\Big( f(x)-\mathrm{e}^{\mathrm{i}\theta (x,y)}f(y)\Big)\overline{\Big(g(x)-\mathrm{e}^{\mathrm{i}\theta (x,y)}g(y)\Big)} + \sum_{x \in X} f(x)\overline{g(x)}v(x)m(x).$$
\end{prop}
\begin{proof}
We start by proving the closability for $v_{+}$.  Define the form
$$
Q^{\max}_{v_{+},\theta}:\ell^{2}(X,m)\longrightarrow [0,\infty]
$$
by
\begin{align*}
    Q^{\max}_{v_{+},\theta}(f)=\frac{1}{2}\sum_{x,y \in X} b(x,y) |f(x)-\e^{\mathrm{i} \theta(x,y) }f(y)|^{2} + \sum_{x\in X} v_{+}(x)|f(x)|^{2} m(x).
\end{align*}
In order to show that $Q^{(c)}_{v_{+},\theta}$ is closable, it suffices to demonstrate that $Q^{\max}_{v_{+},\theta}$ is  lower semi-continuous. This  is a consequence of Fatou's lemma.\\
Now, for $v \in \mathcal{B}_{\theta}$, {there is $\eps>0$ and $C>0$  such that for  $C'> C$, we obtain the inequalities
$$\eps Q_{v_+,\theta}(f) + (C'-C)\|f\|^2 \leq Q^{(c)}_{v,\theta}(f) + C'\|f\|^2 \leq Q_{v_+,\theta}(f) + C'\|f\|^2$$
for all $f \in C_c(X)$. These inequalities show that both form norms have the same Cauchy sequences. Thus, the closability of $Q^{(c)}_{v_+,\theta}$ implies the closability of $Q^{(c)}_{v,\theta}$ and the equality $D(Q_{v,\theta}) = D(Q_{v_{+},\theta})$.} For the statement on the action of the form, let us first note that
$$Q^{\rm max}_{v ,\theta}(f) = \frac{1}{2}\sum_{x,y \in X} b(x,y) |f(x)-\e^{\mathrm{i} \theta(x,y)}f(y)|^{2} + \sum_{x\in X} v(x)|f(x)|^{2} m(x)$$
is well defined for all $f\in D(Q_{v,\theta})$, i.e., $Q^{\rm max}_{v_+ ,\theta}(f)< \infty$ and $q_{v_-}(f) < \infty$. To see this, pick a sequence of compactly supported functions $(f_n)$ converging to $f$ with respect to the form norm induced by $Q_{v_+,\theta}$. We then obtain by Fatou's lemma and  by $v \in \mathcal{B}_{\theta}$
$$Q^{\rm max}_{v_+ ,\theta}(f) \leq \liminf_{n \to \infty} Q_{v_+ ,\theta}(f_n) = Q_{v_+ ,\theta}(f)$$
and
$$q_{v_-}(f) \leq \liminf_{n\to \infty} q_{v_-}(f_n) \leq \liminf_{n\to \infty} (1-\varepsilon)Q_{v_+,\theta}(f_n) + C\|f_n\|^2 = (1-\varepsilon)Q_{v_+,\theta}(f) + C\|f\|^2.$$
Altogether, the above, Fatou's lemma and $Q^{\rm max}_{v,\theta}$ being a quadratic form implies
\begin{align*}
|Q^{\rm max}_{v ,\theta}(f) - Q_{v ,\theta}(f)|^{1/2} & = \lim_{n \to \infty} |Q^{\rm max}_{v ,\theta}(f) - Q^{\rm max}_{v ,\theta}(f_n)|^{1/2}\\
& \leq \liminf_{n \to \infty} |Q^{\rm max}_{v_{+} ,\theta}(f) - Q^{\rm max}_{v_+ ,\theta}(f_n)|^{1/2} + \liminf_{n \to \infty} |q_{v_-}(f) - q_{v_-}(f_n)|^{1/2} \\
& \leq \liminf_{n \to \infty} Q^{\rm max}_{v_{+} ,\theta}(f-f_n)^{1/2} +  \liminf_{n \to \infty}q_{v_-}(f-f_n)^{1/2} \\
& \leq  \liminf_{n,m \to \infty} Q^{\rm max}_{v_{+} ,\theta}(f_m - f_n)^{1/2} \\ & \qquad + \liminf_{n,m \to \infty} \left((1-\eps)Q^{\rm max}_{v_{+} ,\theta}(f_m - f_n) + C\| f_m - f_n\|^2 \right)^{1/2}.
\end{align*}
As $(f_n)$ is a Cauchy-sequence with respect to the form norm of $Q_{v_+,\theta}$, these computations show the claim.
\end{proof}

In fact, we are going to prove the following generalization of Proposition \ref{p:closable} later on, which will not be used in the sequel of this section, but certainly it is of an independent interest.
The proof, given in Section~\ref{s:closable_proof}, works by an approximation argument, cutting off the negative parts of the potentials and employing
a Feynman-Kac-It\^o formula for potentials that are bounded from below (and, thus, belong to $\mathcal{B}_{0}$).  Eventually, we will use this result to show a  Feynman-Kac-It\^o formula for potentials in $\mathcal{A}_{0}$.

\begin{thm}\label{t:closable} For every  $v \in \mathcal{A}_{0}$ the form  $Q^{(c)}_{v,\theta}$ is semi-bounded and closable.
\end{thm}

For non-magnetic forms, i.e., $\theta=0$, even the converse is true.

\begin{coro} The form $Q^{(c)}_{v,0}$ is semi-bounded and closable if and only if $v\in \mathcal{A}_{0}$.
\end{coro}
\begin{proof} The ``if'' follows directly from Theorem~\ref{t:closable} and the ``only if'' follows as $Q_{v,0}^{(c)}$ is not semi-bounded if $v$ is not in $\mathcal{A}_{0}$, by Lemma \ref{l:potentials}.
\end{proof}

We proceed by giving further criteria for $Q^{(c)}_{v,\theta}$ being closable and for the operator $L_{v,\theta}$ associated to its closure being a restriction of $\ow L_{v,\theta}$. To this end the condition $\ow L_{v,\theta}[{C}_c(X)] \subseteq \ell^{2}(X,m)$ plays a role. We put this condition into perspective which is based on an observation made in \cite{KL} for the Dirichlet form case.}

\begin{lemma} \label{finiteness assumption} The following assertions are equivalent:
\begin{itemize}
\item[(i)] $\ow L_{v,\theta}[{C}_c(X)] \subseteq \ell^{2}(X,m)$.
\item[(ii)] $\ow L_{0,0}[{C}_c(X)] \subseteq \ell^{2}(X,m)$
\item[(iii)] For all $x \in X$ the function $X \to [0,\infty)$, $y \mapsto b(x,y)/m(y)$ belongs to $\ell^2(X,m)$.
\end{itemize}
If one of the above is satisfied, then $\ell^2(X,m) \subseteq \ow{F}(X)$. { Furthermore, the assertions are implied by local finiteness of the graph $b$ or $m\ge C$ for some $C>0$.}
\end{lemma}
\begin{proof} The proof follows from a straightforward computation, see e.g. \cite[Proposition~3.3]{KL} { and \cite[Lemma~2.3]{semic}} for details.
\end{proof}

We can now state the theorem about the action of the operators.

\begin{thm} \label{operator}
Suppose one of the following conditions holds.
\begin{itemize}
\item[(a)] $v \in \mathcal{B}_{\theta}$.
\item[(b)] $v \in \mathcal{A}_\theta$ and $\ow{L}_{v,\theta}[{C}_c(X)] \subseteq  \ell^2(X,m)$.
\item[(c)] $v \in \mathcal{A}_\theta$ and $(X,b)$ is locally finite.
\end{itemize}
Then $Q^{(c)}_{v,\theta}$ is semi-bounded and closable and the corresponding operator is a restriction of $\ow{L}_{v,\theta}$.
\end{thm}
\begin{proof} Clearly assumption (c) implies (b), hence it suffices to show the statement under assumption (a) and (b). Let us assume (a). As seen in Proposition \ref{p:closable}, the form $Q^{(c)}_{v,\theta}$ is semi-bounded, closable and satisfies $D(Q_{v,\theta}) =  D(Q_{v_+,\theta})$. We will now show  $D(Q_{v_+,\theta}) \subseteq \ow{F}(X)$. The inclusion  $D(Q_{v_+,\theta}) \subseteq \ow{F}(X)$ together with the action of $Q_{v,\theta}$ (Proposition~\ref{p:closable}) and Green's formula (Lemma~\ref{Greens formula}) would imply
$$\as{L_{v,\theta}f,g} = Q_{v,\theta}(f,g) = \sum_{x \in X}\ow{L}_{v,\theta}f(x) \overline{g(x)}m(x) $$
for all $f \in D(L_{v,\theta})$ and $g \in C_c(X)$. So, showing $D(Q_{v_+,\theta}) \subseteq \ow{F}(X)$ would prove the claim. Thus, let $f \in D(Q_{v_+,\theta})$ be given. We  estimate
\begin{align*}
\sum_{y \in X}b(x,y)|f(y)| &\leq \sum_{y \in X}b(x,y)|f(x) - \mathrm{e}^{\mathrm{i}\theta(x,y)}f(y)| + \sum_{y \in X}b(x,y)|f(x)| \\
& \leq {\rm deg_1}(x)^{1/2} \left(\sum_{y \in X}b(x,y)|f(x) - \mathrm{e}^{\mathrm{i}\theta(x,y)}f(y)|^2\right)^{1/2} + {\rm deg_1}(x) |f(x)|,
\end{align*}
{where $\deg_{1}(x)=\sum_{y\in X}b(x,y)$ is finite by assumption on the graph $b$ and the form expression  is finite by Proposition \ref{p:closable}. Hence, $\sum_{y \in X}b(x,y)|f(y)|<\infty$ which implies $f\in \ow F(X)$.}

Next, we assume (b) holds. Then $Q^{(c)}_{v,\theta}$ is semi-bounded and closable by the Friedrich's extension theorem.  Let $f \in D(L_{v,\theta})$ be given and $(f_n)$ be a sequence of compactly supported functions converging to $f$ in the form norm. Then, for all $g \in C_c(X)$, we obtain by definition of $L_{v,\theta}$ and Green's formula (Lemma~\ref{Greens formula})
\begin{align*}
\as{L_{v,\theta}f,g} &= Q_{v,\theta}(f,g) \\
& = \lim_{n\to \infty} Q^{(c)}_{v,\theta}(f_n,g)\\
&= \lim_{n\to \infty} \sum_{x \in V} \ow{L}_{v,\theta} f_n (x) \overline{g(x)}m(x).
\end{align*}
As $g$ is compactly supported,  it suffices to show the pointwise convergence of $\ow{L}_{v,\theta} f_n$ towards $\ow{L}_{v,\theta} f$ to prove the claim. For this it is sufficient to show the convergence
$$\sum_{y \in X}b(x,y)f_n(y)\e^{\mathrm{i}\theta(x,y)} \to \sum_{y \in X}b(x,y)f(y)\e^{\mathrm{i}\theta(x,y)},\;n\to\infty, $$
for each $x \in X$. This can be deduced from
$$\sum_{y \in X}b(x,y)|f_n(y)-f(y)| \leq \left(\sum_{y\in X}\frac{b(x,y)^2}{m(y)}\right)^{1/2}\left(\sum_{y\in X}|f_n(y)-f(y)|^2m(y)\right)^{1/2},$$
where the finiteness of the first factor of the right-hand side follows from the characterization of the assumption $\ow{L}_{v,\theta}[{C}_c(X)] \subseteq  \ell^2(X,m)$ in  Lemma~\ref{finiteness assumption} and finiteness of the second factor follows from $f, f_{n}\in\ell^{2}(X,m)$.
\end{proof}


\subsection{Uniqueness of semi-bounded self-adjoint restrictions}\label{s:uniqueness}{
In this section we present uniqueness results for self-adjoint operators that are restrictions of $\ow L_{v,\theta}$. It is intended to complement the operator theoretic picture and extend previous results in this direction to our much more general situation. However, these result will not be needed for the Feynman-Kac-It\^o formula.

A classical approach to uniqueness results of self-adjoint operators is the concept of essential self-adjointness. However, the notion of essential self-adjointness of $\ow L_{v,\theta}$ only makes sense if $\ow L_{v,\theta}[{C}_c(X)]\subseteq \ell^{2}(X,m)$. Nevertheless, \emph{in general}, we can} still ask for the uniqueness of semi-bounded self-adjoint restrictions of $\ow L_{v,\theta}$ on $\ell^{2}(X,m)$ in an appropriate sense. (Precisely, we ask whether there is a unique dense subspace $D$ of $\ell^{2}(X,m)$ such that the restriction of $\ow L_{v,\theta}$ to $D$ is a semi-bounded self-adjoint operator.) 

{Let us mention that, in general, it is not clear whether $\ow L_{v,\theta}$ has a self-adjoint restriction to $\ell^{2}(X,m)$ at all.

We start by presenting an abstract criterion for uniqueness in case of existence of self-adjoint restrictions based on uniqueness of the solutions of $(\ow L_{v,\theta} - \lm)u=0$. This is complemented by two two conditions each of which ensuring existence. This is the content of , Proposition~\ref{p:criterion}.

Afterwards, we give two criteria under which the assumptions of Proposition~\ref{p:criterion} are met. The first one,  Theorem~\ref{t:essSA1}, essentially makes an assumption on the underlying weighted graph as a measure space, and the second one, Theorem~\ref{t:essSA2}, makes an assumption on the weighted graph as a metric space.

We start with the abstract criteria for uniqueness and existence.}

\begin{prop}\label{p:criterion}  Assume there exists some constant $C\in \IR$ such that for all $\lambda < C$, every solution $u\in\ow F(X)\cap \ell^{2}(X,m)$ of $(\ow L_{v,\theta} - \lm)u=0$ satisfies $u\equiv0$. Then $\ow L_{v,\theta}$ has at most one semi-bounded self-adjoint restriction on $\ell^{2}(X,m)$. Furthermore, the following holds:
\begin{itemize}
\item[(a)] If, additionally,  $v \in \mathcal{B}_{\theta}$, then $\ow{L}_{v,\theta}$ has a unique semi-bounded self-adjoint restriction.
\item[(b)] If, additionally,  $v \in \mathcal{A}_{\theta}$ and $\ow{L}_{v,\theta}[{C}_c(X)]\subseteq \ell^{2}(X,m)$, then $\ow{L}_{v,\theta}\vert_{C_{c}(X)}$ is essentially self-adjoint.
\end{itemize}
\end{prop}

\begin{proof} Suppose there are two such restrictions $L_{1}$ and $L_{2}$ on $\ell^{2}(X,m)$ which do not coincide. Let $C$ be a common lower bound of $L_1$ and $L_2$. Then, their resolvents $(L_{1}-\lm)^{-1}$ and $(L_{2}-\lm)^{-1}$ are different for $\lm<C$. Hence, we infer
$$
u=((L_{1}-\lm)^{-1}-(L_{2}-\lm)^{-1})\ph\neq0\>\> \text{for some $\ph\in C_{c}(X)$.}
$$
As $L_{1}$ and $L_{2}$ are both restrictions of $\ow L_{v,\theta}$, we have $(\ow L_{v,\theta}-\lm)u=\ph-\ph=0$ and get a contradiction.

Under the additional assumption in (a) the existence of a semi-bounded self-adjoint restrictions follows from Theorem \ref{operator}.

For the statement under the assumptions of (b) assume $\ow L_{v,\theta}[{C}_c(X)]\subseteq \ell^{2}(X,m)$. Let $L_{\rm min} = \ow L_{v,\theta}|_{C_c(X)}$ and $L_{\rm max} = L_{\rm min}^*$ its adjoint. It suffices to show that $L_{\rm max}$ is self-adjoint. From Lemma \ref{finiteness assumption} we infer $\ell^2(X,m) \subseteq  \tilde{F}$. This allows the application of Green's formula (Lemma~\ref{Greens formula}), i.e.,
$$\as{u,\ow{L}_{v,\theta}f} = \as{\ow{L}_{v,\theta}u,f} $$
for all $u\in \ell^2(X,m)$, such that $\ow{L}_{v,\theta}u \in \ell^2(X,m)$ and all $f\in C_c(X)$. This shows that $L_{\rm max}$ is a restriction of $\ow{L}_{v,\theta}$  with domain
$$
D(L_{\rm max}) = \{u\in \ell^2(X,m)\>|\>\> \ow{L}_{v,\theta}u\in \ell^2(X,m)\}.
$$
Now let $L_{v,\theta}$ be the self-adjoint semi-bounded operator associated with the closure of $Q^{(c)}_{v,\theta}$. By Theorem \ref{operator}, $L_{v,\theta}$ is a restriction of $\ow{L}_{v,\theta}$, satisfying $D(L_{v,\theta}) \subseteq D(L_{\rm max})$. Therefore, it suffices to show the other inclusion. Let $u \in D(L_{\rm max})$ be given and let $w = (L_{v,\theta}-\lambda)^{-1}(\ow{L}_{v,\theta} - \lambda)u \in D(L_{v,\theta}).$ We obtain $(\ow{L}_{v,\theta}-\lambda)(w-u) = 0,$ which implies $u = w\in D(L_{v,\theta})$ by our assumptions.
\end{proof}

The first criterion for uniqueness is based on a result from \cite{KL} for $\theta=0$. This was later generalized to locally finite magnetic operators in \cite{Gol}. The result below stands somewhat skew to the one of \cite{Gol}: In \cite{Gol} no assumption on the semi-boundedness of the quadratic form is made, whereas we do not assume local finiteness.

\begin{thm}\label{t:essSA1}\emph{(Uniqueness - measure space criterion)} Assume that for some $\al\in\R$ and all infinite paths $(x_{n})^{\infty}_{n=0}$ one has
\begin{align*}
    \sum^{\infty}_{n=1}m(x_{n})\prod_{j=0}^{n-1} \left(1+\frac{v(x_{j})-\al}{\deg_m(x_{j})}\right)^{2}=\infty.
\end{align*}
Then the following holds:
\begin{itemize}
\item[(a)] If, additionally, $v\in \mathcal{B}_{\theta}$, then ${\ow L_{v,\theta}}$ has a unique semi-bounded self-adjoint restriction.
\item[(b)] If, additionally, $v \in \mathcal{A}_{\theta}$ and $\ow{L}_{v,\theta}[{C}_c(X)]\subseteq \ell^{2}(X,m)$, then $\ow{L}_{v,\theta}\vert_{C_{c}(X)}$ is essentially self-adjoint.
\end{itemize}
\end{thm}

\begin{proof} As $Q_{v,\theta}^{(c)}$ is bounded below by some constant $C$, we infer $\deg_{m}+v-\lm >0$ for all $\lm<C$. Thus, if the sums in the assumption diverge for a particular $\al$, then there is $\lm_{0}< -(|C|+|\al|)$ such that  these sums diverge for all $\lm < \lm_{0}$.   Let $u\in\ell^{2}(X,m)\cap\ow F(X)$  be a solution to  the  equation $(\ow L_{v,\theta}-\lm)u=0$ for some $\lm<\lm_{0}$. Then, one easily gets
$$|u(x)|\leq\left(\frac{1}{m(x)}\sum_{y\in X}b(x,y)|u(y)|\right)|\deg_{m}(x)+v(x)-\lm|^{-1},$$
for all $x \in V$. Suppose $u \not\equiv 0$, i.e., there exists an $x_0\in X$ such that $u(x_0) \neq 0$. By the above inequality there is an $x_1\sim x_0$ with
$$|u(x_1)|\geq \left|1+\frac{v(x_0)-\lm}{\deg_{m}(x_0)}\right||u(x_0)|.$$
 Continuing this procedure, we may inductively choose an infinite path $(x_{n})$ which satisfies
$$|u(x_n)| \geq \prod_{i=0}^{n-1}\left|1+\frac{v(x_i)-\lm}{\deg_{m}(x_i)}\right||u(x_0)|.$$
Therefore, we obtain
$$\|u\|^{2}\ge\sum_{n=1}^{\infty}|u(x_{n})|^{2}m(x_{n})\ge \sum_{n=1}^{\infty}m(x_{n})|u(x_{0})|^{2}\prod_{j=0}^{n-1} \left|1+\frac{v(x_{j})-\lm}{\deg_{m}(x_{j})}\right|.$$
This implies $u(x_{0})=0$ by the assumption. As this contradicts $u(x_0) \neq 0$, we  conclude $u\equiv0$. Thus, the statement follows from Proposition~\ref{p:criterion}.
\end{proof}

As we have already remarked, the second criterion is going to deal with the completeness of the weighted graph with respect to some appropriate metric structure.

\begin{definition}A (pseudo) metric $d$ on  $X$ is called a \emph{path (pseudo) metric} for the graph $b$, if there is a map $\si:X\times X\to[0,\infty)$ with the properties $\{\sigma=0\}\subseteq\{b=0\}$ and
\begin{align*}
  d(x,y)=\inf_{x=x_{0}\sim\ldots\sim x_{n}=y}\sum_{j=1}^{n}\si(x_{j-1},x_{j}),\quad\mbox{for all $x,y\in X$}.
\end{align*}
A (pseudo) metric $d$ is called \emph{intrinsic} with respect to $(X,b,m)$, if
\begin{align*}
    \sum_{y\in X}b(x,y)d(x,y)^{2}\leq m(x),\quad\mbox{for all $x\in X$}.
\end{align*}
\end{definition}

We remark that the above definition of intrinsic metrics is adapted to our situation from the abstract Dirichlet space setting of \cite{FLW}. Furthermore, any weighted graph admits an intrinsic metric. For example, one can take the path metric with weights $\si(x,y)=(\deg_m(x)\wedge\deg_m(y))^{-\frac{1}{2}}$ for $ x\sim y$.

Next, we present a  result which has also been suggested to us by O. Milatovic in a private communication. Earlier results of this type for magnetic operators in the continuum with similar kinds of proofs already appeared in  \cite{CdVT}.

\begin{thm}\label{t:essSA2}\emph{(Uniqueness - metric space criterion)} Let $d$ be an intrinsic pseudo metric with respect to the underlying weighted graph.
\begin{itemize}
  \item [(a)]  Assume  $v \in \mathcal{B}_{\theta}$ and that the metric balls with respect to $d$ are all finite. Then the operator $\ow L_{v,\theta} $ has a unique semi-bounded self-adjoint restriction.
  \item [(b)] Assume   $v \in \mathcal{A}_\theta$ and that the underlying graph is locally finite and $(X,d)$ is a complete path metric space. Then the operator ${\ow L_{v,\theta}\vert}_{{C}_c(X)}$ is essentially self-adjoint.
\end{itemize}
\end{thm}

\begin{remark} (a) Theorem~\ref{t:essSA2} is a generalization of \cite[Corollary~1 and Theorem~2]{HKMW} and \cite[Theorem~1.5]{Mi}. While the first reference does not allow magnetic fields and negative potentials, the second one assumes a uniformly bounded vertex degree, a condition that we will avoid by using the concept of intrinsic metrics. The proof works analogously to \cite{Mi}. We refer also to \cite{MT} for results in this direction.\\
(b) In view of the Kato class being contained in $\mathcal{B}_{0}\subseteq \mathcal{B}_{\theta}\subseteq \mathcal{A}_{\theta}$ (cf. \cite[Theorem~3.1]{peter}), Theorem~\ref{t:essSA2}  can be considered in fact as a weighted-graph analogue of the corresponding result for geodesically complete Riemannian manifolds from \cite{bg2}.
\end{remark}

The proof of Theorem~\ref{t:essSA2} given below, is based on the following ground state transform: For any $f=f_1+\mathrm{i}f_2$ with real-valued $f_1,f_{2}\in \ow{{\F}}(X)$ we define
\begin{align*}
     Q^{(f)}(g,h)=\frac{1}{2} \sum_{x,y\in X} b^{(f)}(x,y)\Big(g(x)-g(y)\Big)\ov{\Big(h(x)-h(y)\Big)},\qquad g,h\in {C}_c(X),
\end{align*}
where  $ b^{(f)} (x,y)$ is defined  for $ x,y\in X $ as
\begin{align*}
b(x,y) \Big(\cos (\te(x,y) ) \big(f_{1}(x)f_{1}(y)+f_{2}(x)f_{2}(y)\big)+\sin( \te(x,y))\big(f_{1}(y)f_{2}(x)-f_{1}(x)f_{2}(y)\big)\Big).
\end{align*}

\begin{prop}\label{mil} Assume $f\in \ow{{\F}}(X)$ and $\lm\in \R$ are such that $(\ow L_{v,\theta}-\lm )f=0$. Then, for all $g\in {C}_c(X)$, one has
\begin{align*}
    Q^{(c)}_{v,\theta}(fg,fg)= Q^{(f)}(g,g)+\lm\|fg\|^{2}.
\end{align*}
\end{prop}
\begin{proof} The proof follows by direct calculation (cf. \cite[Proposition~3.5]{Mi} or \cite[Proposition~3.2]{HK}).
\end{proof}

\begin{proof}[Proof of Theorem~\ref{t:essSA2}]
Let $C$ be such that $q_{v_{-}}(f)\leq (1-\varepsilon) Q_{v_{+},\theta}^{(c)}(f)+C\|f\|^{2}$ for $\in f\in C_{c}(X)$ with $\varepsilon>0 $ in case (a) and $\varepsilon=0$ in the case (b). 
Let  $f\in\ell^{2}(X,m)\cap\ow F(X)$ and  $\lm < -C+1 $ be such that $ (\ow L_{v,\theta}-\lm)f=0$.
We fix some $x_{0}\in X$ and denote the $R$-ball, $R>0$, with respect to $d$ with center $x_{0}$ by $B_{R}$. Let $\eta_{R}:X\to[0,1]$,  be given by
\[
\eta_{R}(x):= 1\wedge\frac{(2R-d(x,x_{0}))_{+}}{R}\>\>\text{ for $x\in X$ }.
\]
By a Hopf-Rinow  type theorem,  \cite[Theorem~A1]{HKMW}  (cf. \cite[ Section~6]{Mi}), the balls are  finite under the metric completeness assumption in (b). Hence,  finiteness of the balls in (a) and (b) implies $\eta_R \in {C}_c (X)$. Then using ${\eta_{R}\vert}_{B_{R}}\equiv1$, the semi-boundedness  of ${Q_{v,\theta}^{(c)}}$ by $\lm + 1$, and Proposition~\ref{mil}, we obtain
\begin{align*}
\|f 1_{B_{R}}\|^{2}&\le\|f \eta_{R}\|^{2}\leq  Q^{(c)}_{v,\te}(f\eta_{R},f\eta_{R}) - \lm\|f \eta_{R}\|^{2} =  Q^{(f)}( \eta_{R},\eta_{R}).
\end{align*}
Employing  the  inequalities
$ b^{(f)}(x,y)\le b(x,y)(|f(x)|^{2}+|f(y)|^{2})$, $(\eta_{R}(x)-\eta_{R}(y))\leq d(x,y)/R$
and the intrinsic metric property, yields
\begin{align*}
\ldots\le \sum_{x\in X} |f(x)|^{2}\sum_{y\in X}b(x,y)(\eta_{R}(x)-\eta_{R}(y))^{2}\leq \frac{1}{R^{2}}\sum_{x\in X} |f(x)|^{2}\sum_{y\in X}b(x,y)d(x,y)^{2}\leq\frac{1}{R^{2}}\|f\|^{2}.
\end{align*}
Letting $R\to\infty$ shows that $\|f\|=0$. Thus, any solution $f$ in $\ell^{2}(X,m)\cap\ow F(X)$ to $(\ow L_{v,\theta}-\lm) f=0$ is trivial. Thus, (a) follows directly  by Proposition~\ref{p:criterion} while for (b) we additionally have to invoke that local finiteness implies  $\ow L_{v,\theta}[{C}_c(X)]\subseteq \ell^{2}(X,m)$.
\end{proof}


\subsection{Semigroup convergence}\label{semi}

We close this section with a result on the convergence of certain geometrically defined restrictions of the semigroups $(\mathrm{e}^{-tL_{v,\theta}})_{t\geq 0}$. This result will be central for the proof of the Feynman-Kac-It\^{o} formula.

We start by introducing some notation that will be useful in the sequel: For any \emph{finite} subset $U\subseteq X$, we denote with slight abuse of notation the restriction of $m$ to $U$ also by $m$ and we define $Q^{(U)}_{v,\theta}$ to be the restriction of $Q^{(c)}_{v,\theta}$ to
$$
\ell^2(U,m)=C_c(U)=C(U).
$$

Here,  the finiteness of $U$ implies that $Q^{(U)}_{v,\theta}$ is automatically closed. Let $L^{(U)}_{v,\theta}$ be the operator corresponding to $Q^{(U)}_{v,\theta}$. We have a canonic inclusion operator
$$
\iota_U: \ell^2(U,m)\hookrightarrow \ell^2(X,m)
$$
which comes from extending functions to zero away from $U$, and its adjoint will be denoted with $\pi_U:=\iota^*_U$.

\begin{definition}\label{exhaus} A sequence $(X_{n})_{n\in\IN}$ of finite sets $X_{n}\subseteq X$ is called an \emph{exhausting sequence} for  $X$, if $X_{n}\subseteq X_{n+1}$ for all $n$ and if $X=\bigcup_{n\in\IN}X_{n}$.
\end{definition}

The following geometric approximation is based on the Mosco convergence of the quadratic forms.

\begin{prop} \label{approx1}
Suppose $Q^{(c)}_{v,\theta}$ is semi-bounded and closable and let $(X_n)_{n\in\IN}$ be an exhausting sequence. Then, for all $t\geq 0$, one as
$$\iota_{X_n} \mathrm{e}^{-tL_{v,\theta}^{(X_n)}}\pi_{X_n} \to \mathrm{e}^{-tL_{v,\theta}} \text{ strongly in } \ell^2(X,m) \mbox{ as }n\to\infty.$$
\end{prop}
\begin{proof} By Theorem~\ref{mosco.char} it suffices to show that the forms $Q^{(X_n)}_{v,\theta}$ converge to $Q_{v,\theta}$ as $n\to\infty$ in the generalized Mosco sense. Part (a) of Definition~\ref{mosco} follows from the closedness of $Q_{v,\theta}$ while part (b) is due to the fact that $C_c(X)$ is a core for $Q_{v,\theta}$ by definition.
 \end{proof}

%

\section{Stochastic processes on discrete sets}\label{proc}

Let us introduce the necessary probabilistic framework. {That is we construct a Markov process. Later in Section~\ref{haup} we show that this Markov process appears in the Feynman-Kac-It\^o formula and is therefore related to the semigroups of the operators considered above. Furthermore, we construct a discrete stochastic line integral with respect to this process.}

We take a discrete time Markov chain $(Y_n)_{n\in\IN}$ with state space $X$ which satisfies
$$\mathbb{P}\left(Y_n = x \left| Y_{n-1} = y\right)\right. = \frac{b(x,y)}{{\rm deg}_1(y)}\>\>\text{ for all $n\ge1$},$$
where in the following $(\Omega,\mathcal{F},\mathbb{P})$ is some fixed probability space, $\deg_{1}(x)=\sum_{y\in X}b(x,y)$, $x\in X$,  and $\N=\{0,1,2,\ldots\}$. Let $(\xi_n)_{n\in\IN}$ be a sequence of independent exponentially distributed random variables of parameter $1$ which are also independent of $(Y_n)_{n\in\IN}$. For $n\geq 1$, we define the sequence of stopping times 
$$J_n: = \frac{1}{{\rm deg}_m(Y_{n-1})}\xi_n, \quad\tau_n: = J_1 + \cdots + J_n,$$
with the convention $\tau_0:= 0$. { Furthermore, we define the stopping time
$$
\tau: = \sup_{n\in\IN} \tau_n:\Omega\longrightarrow [0,\infty],
$$
where obviously $\tau>0$ is satisfied $\mathbb{P}$-almost surely.}

With these preparations, we define the jump process
\[
\mathbb{X}:[0,\tau)\times \Omega\longrightarrow X,\quad{\mathbb{X}\vert}_{[\tau_n,\tau_{n+1})\times \Omega}\>:=Y_n\quad\mbox{ for all $n\in\IN$}.
\]
Note that $\mathbb{X}$ is maximally defined and that the $\tau_n$\rq{}s are precisely the jump times of $\mathbb{X}$. If $\mathbb{P}_x:=\mathbb{P}(\bullet\mid\mathbb{X}_0=x)$, and if $\mathcal{F}_{*}$ denotes the filtration $\mathcal{F}_t=\sigma(\mathbb{X}_s\mid s\leq t)$, $t\ge0$, corresponding to $\mathbb{X}$, then the tuple
\begin{align*}
 (\Omega,\mathcal{F},\mathcal{F}_*,\mathbb{X} ,(\mathbb{P}_x)_{x\in X})
\end{align*}
is a (reversible) strong Markov process (see for example Theorem 6.5.4 in \cite{No} for a proof).

Let us denote the number of jumps of $\mathbb{X}$ until $t$ by $N(t)$, i.e.,
$$N(t) = \sup\{n\in\IN\mid \tau_n \leq t \}.$$

The following definitions will be central for this paper. We define two random variables by
\[
\int_0^t \theta( \Id \mathbb{X}_s):= \sum_{n= 1}^{N(t)}\theta(\mathbb{X}_{\tau_{n-1}}, \mathbb{X}_{\tau_{n}}): \{t<\tau \}\longrightarrow \IR
\]
and
\[
\ISS_t(v,\theta|\mathbb{X}):=\mathrm{i}\int_0^t \theta( \Id \mathbb{X}_s)-\int_0^t v( \mathbb{X}_s)\Id s:\{t<\tau\}\longrightarrow \IC.
\]
{In particular, $\ISS_t(v,0|\mathbb{X})$ can be seen as} the usual additive Feynman-Kac functional
$\ISS_t(v,0|\mathbb{X})=-\int_0^t v( \mathbb{X}_s)\Id s$.

The well-definedness of $\int_0^t \theta( \Id \mathbb{X}_s)$ and $\int_0^t v( \mathbb{X}_s)\Id s$ (and thus of $\ISS_t(v,\theta|\mathbb{X})$) follows from the simple observation $\{N(t) < \infty\}=\{t < \tau\}$.

Furthermore, it is easily seen that the processes
\[
\int_0^{\bullet} \theta( \Id \mathbb{X}_s):[0,\tau)\times \Omega\longrightarrow  \IR, \quad\ISS_{\bullet}(v,\theta|\mathbb{X}):[0,\tau)\times \Omega\longrightarrow \IC
\]
are $\IFF_*$-semimartingales under $\mathbb{P}_x$ with lifetime $\tau$, which motivates the following definition.

\begin{definition} The process $\int_0^{\bullet} \theta( \Id \mathbb{X}_s)$ is called the \emph{stochastic line integral of $\theta$ along $X$}, and $\ISS_{\bullet}(v,\theta;\mathbb{X})$ is called the \emph{Euclidean action corresponding to $\theta$ and $v$.}
\end{definition}

Here, the notions \lq\lq{}line integral\rq\rq{} and \lq\lq{}Euclidean action\rq\rq{} are both motivated from the manifold setting \cite{Em}, where in the first case $\theta$ is interpreted as a $1$-form on the graph $X$. We refer the reader to \cite{Mi} for a justification of the latter geometric interpretation.

{
Let us end this section by putting the process $\mathbb{X}$ into perspective.

\begin{remark}It is certainly well known that the constructed process $\mathbb{X}$ is related to semigroup $\e^{-tL}$ of the operator $L=L_{0,0}$ introduced in the previous section via the formula
\begin{align*} \e^{-tL}f(x)  =\mathbb{E}_x\left[1_{\{t<\tau\}} f(\mathbb{X}_t)\right]. \end{align*}
In any case, this formula is a special case of the Feynman-Kac-It\^o formula proven in the next section.\\
This formula has a simple but nevertheless important consequence, namely, one has
\begin{align*}
\e^{-t L}(x,y)m(y) =\mathbb{P}_x(\mathbb{X}_t=y)\>\>\text{ for all $t>0$, $x,y\in X$},
\end{align*}
where the kernel $\e^{-t L}(x,y)$, $x,y\in X$, of $\e^{-tL}$ exists due to discreteness of the space. \\
In particular, it follows that the process $(\Omega,\mathcal{F},\mathcal{F}_*,\mathbb{X} ,(\mathbb{P}_x)_{x\in X})$ is non-explosive, i.e.,
\[
\mathbb{P}_x(\tau=\infty)=1\quad\text{ for all $x\in X$},
\]
if and only if one has
\begin{align*}
\sum_{y\in X }\e^{-t L}(x,y)m(y)=1\quad\text{ for all $t\geq 0, x\in X$.}
\end{align*}
This follows from combining the formula $\e^{-t L}(x,y)m(y) =\mathbb{P}_x(\mathbb{X}_t=y)$ with $
\{\tau=\infty\}=\bigcap_{n\in\IN} \{\tau> n\}$ keeping $\mathbb{P}_x(\tau>0)=1$ in mind.\\
In summary, the Dirichlet form $Q=Q_{0,0}$ is stochastically complete, i.e., $\mathrm{e}^{-tL}1=1$, if and only if the process is non-explosive, i.e., $\mathbb{P}_x(\tau=\infty)=1$, $x\in X$, a well-known fact which is found already in \cite[Exercise~4.5.1]{FOT}. In case the underlying process is non-explosive, some of the considerations below become somewhat simpler, nevertheless, there are many  graphs where explosion can occur, see e.g. \cite{KL,Woj1,Woj2}.
\end{remark}
}



\section{The Feynman-Kac-It\^{o} Formula}\label{haup}
\subsection{Statement}

The following theorem is the main result of this paper.

\begin{thm}\label{main}\emph{(Feynman-Kac-It\^{o} formula)}
Let $v\in \mathcal{A}_{0}$.
Then for any $f \in \ell^2(X,m)$, $t\geq 0$ and $x\in X$, one has
\begin{align}\label{fki}\tag{FKI}
\e^{-tL_{v,\theta}}f(x) = \mathbb{E}_x\left[1_{\{t<\tau\}}\e^{\ISS_t(v,\theta|\mathbb{X})} f(\mathbb{X}_t)\right].
\end{align}
\end{thm}

{
The rest of the section is dedicated  to the proof of Theorem~\ref{main}. The proof is divided into several parts:
\begin{itemize}
  \item [] \emph{Part 1:} We prove (\ref{fki}) for finite subgraphs in Theorem~\ref{p:finite}. Here, we use the explicit form of the process $\mathbb{X}$.
  \item []\emph{Part 2:} We show (\ref{fki}) for the case where $Q_{v,\theta}^{(c)}$ and  $Q_{v,0}^{(c)}$ are both closable, Theorem~\ref{p:FKI_closable}. Here, we use that their closures can be well approximated by restrictions to finite subgraphs, see  Proposition~\ref{approx1}.
  \item []\emph{Part 3:} Finally, we show that the forms $Q_{v,\theta}^{(c)}$ and  $Q_{v,0}^{(c)}$ are closable for $v\in \mathcal{A}_{0}$, Theorem~\ref{t:closable} proven in Section~\ref{s:closable_proof}. Here, we use (\ref{fki}) for  potentials whose negative part is bounded (a case which is included in Theorem~\ref{p:FKI_closable} by Proposition~\ref{p:closable}).
\end{itemize}
 }

{
\begin{remark}
(a) It should be noted that we make \emph{no assumptions} on the underlying weighted graph, the magnetic potential $\theta$ and the positive part $v_{+}$ of  $v$. The only assumption on $v_{-}$ is semi-boundedness of the non-magnetic form. We believe that this setting should actually cover all possible applications.

(b) As we have already remarked in the strategy of the proof above, we are actually going to prove the following fact in Theorem~\ref{p:FKI_closable} below: \emph{Formula (\ref{fki}) holds true, if $Q_{v,0}^{(c)}$ and $Q_{v,\theta}^{(c)}$ are closable and semi-bounded.} The latter statement is slightly more general than Theorem~\ref{main}. However, we believe that Theorem~\ref{p:FKI_closable} itself is not of any practical importance, as there is no general machinery to check its assumptions on $v$ directly (whereas $v\in\mathcal{A}_{0}$ can typically checked much more directly; cf. Example~\ref{kkato} and~\ref{ex:deg}). This is the motivation for declaring Theorem~\ref{main} to be our main result. These observations are fully reflected by the fact that the actual derivation of Theorem~\ref{main} from Theorem~\ref{p:FKI_closable} requires some considerable extra work.
\end{remark}
}

\subsection{Proof for finite subgraphs}

For a  finite subset $U\subset X$, we recall the notation from Section~\ref{semi}  and let
\[
\tau_U:=\left.\inf\{s\geq 0\right| \ \mathbb{X}_s\in X\setminus U\}
\]
be the first exit time of $\mathbb{X}$ from $U$, which is a $\IFF_*$-stopping time. The goal of this subsection is to prove the following proposition, which is the main tool in the proof of Theorem~\ref{main}, but is in fact of an independent interest (see also the proof of Proposition~\ref{mon} below). Here, it should again be noted that in view of the finiteness of $U$, the potentials may be arbitrary.

\begin{thm}\label{p:finite} Let $U\subseteq X$ be finite. Then for all $f\in \ell^{2}(U,m)$, $x\in U$, $t\ge0$, one has
\begin{align*}
    \e^{-tL^{(U)}_{v,\theta}}f(x)=\mathbb{E}_x\left[1_{\{t<\tau_U\}} \e^{\ISS_t(v,\theta|\mathbb{X})} f(\mathbb{X}_t)\right].
\end{align*}
\end{thm}

The proof of the proposition above is based on three auxiliary lemmas.

\begin{lemma}\label{l:semigroup} Let $U \subseteq X$ be finite. Then, $(T_t(v,\theta,U))_{t\ge0}$
defined for $f\in\ell^{2}(U,m)$ by
\begin{align*}
T_t(v,\theta,U)f(x):=\mathbb{E}_x\left[1_{\{t<\tau_U\}} \e^{\ISS_t(v,\theta|\mathbb{X})} f(\mathbb{X}_t)\right],\quad x\in U,t\ge0,
\end{align*}
is a strongly continuous semigroup of bounded operators on $\ell^{2}(U,m)$.
\end{lemma}

\begin{proof} The asserted boundedness is trivial and the semigroup property follows from the strong Markov property of $\mathbb{X}$. By the semigroup property it is enough to check strong continuity at $t=0$, which can be easily checked using the boundedness of the integrand and the right continuity of $\mathbb{X}$.
\end{proof}

\begin{lemma} \label{error} Let $f \in {C}_c(X)$, $t>0$, and let the function $\varphi_{t,f}:X\to \IC$ be defined by
$$\varphi_{t,f}(x) := \frac{1}{t}\mathbb{E}_x\left[1_{\{2\leq N(t)<\infty\}} f(\mathbb{X}_t) \right]. $$
Then, for all $x\in X$, one has $\varphi_{t,f}(x)\to 0$ as $t\searrow 0$.
\end{lemma}
\begin{proof}
As $f$ is bounded, it suffices to show
\begin{gather}\frac{1}{t}\mathbb{P}_x(N(t)\geq 2) = \frac{1-\mathbb{P}_x(N(t)= 0)}{t} - \frac{\mathbb{P}_x(N(t)= 1)}{t} \to 0, \quad\text{ as }t \searrow 0 . \label{errorsum}\end{gather}

From the considerations of Section~\ref{proc} we derive
\[
\mathbb{P}_x(N(t)= 0)=\mathbb{P}_x(t<\tau_{1}) = \mathbb{P}_x(\deg_{m}(x)t< \xi_{1})= \e^{-{\rm deg}_{m}(x)t}.
\]
The first summand of the right hand side of \eqref{errorsum} tends to ${\rm deg}_{m}(x)$ as $t \searrow 0$.
For determining the second summand, let us compute
\begin{align*}
\mathbb{P}_x(N(t)= 1)
=& \sum_{y \in X}  \mathbb{P}_{x}(N(t)= 1,\mathbb{X}_{\tau _1} =y)\\
=& \sum_{y \in X}  \mathbb{P}_x( N(t)= 1\mid \mathbb{X}_{\tau _1} = y) \mathbb{P}_x(\mathbb{X}_{\tau _1} = y) \\
=&  \sum_{y \in X, {\rm deg}_{m}(x)\neq {\rm deg}_{m}(y)}  \frac{{\rm deg}_{m}(x)}{{\rm deg}_{m}(x)-{\rm deg}_{m}(y)} \left[\e^{-t{\rm deg}_{m}(y)} - \e^{-t{\rm deg}_{m}(x)}\right] \frac{b(x,y)}{{\rm deg}_{1}(x)}\\
&+ \sum_{y \in X, {\rm deg}_{1}(x)={\rm deg}_{1}(y)}  \left[t{\rm deg}_{m}(x)\e^{-t{\rm deg}_{m}(x)}\right] \frac{b(x,y)}{{\rm deg}_{1}(x)}.
\end{align*}
{The last equality is a consequence of the following two observations: First, the equality $ \mathbb{P}_x(\mathbb{X}_{\tau _1} = y)={b(x,y)}/{{\rm deg}_{1}(x)}$  with $\deg_{1}(x)=\sum_{y\in X}b(x,y)$ is a direct consequence of the construction of $\mathbb{X}$. Secondly, using the notation of Section~\ref{proc}, we observe
\begin{align*}
    \mathbb{P}_x( N(t)= 1\mid \mathbb{X}_{\tau _1} = y)
    &=\mathbb{P}( N(t)= 1\mid \mathbb{X}_{\tau _1} = y,\mathbb{X}_{0}=x)\\
    &=\mathbb{P}(J_{1}\leq t< J_{1}+J_{2} \mid Y_1 = y,Y_{0}=x)\\
      &=\mathbb{P}\Big(
      \frac{1}{\deg_{m}(Y_{0})}\xi_{1}\leq t<       \frac{1}{\deg_{m}(Y_{0})}\xi_{1}+      \frac{1}{\deg_{m}(Y_{1})}\xi_{2} \mid Y_1 = y,Y_{0}=x\Big)\\
      &=\mathbb{P}\Big(
      \frac{1}{\deg_{m}(x)}\xi_{1}\leq t<       \frac{1}{\deg_{m}(x)}\xi_{1}+      \frac{1}{\deg_{m}(y)}\xi_{2} \Big).
\end{align*}
The last equality follows from the fact that the $Y_{n}$
and $\xi_{n}$ are chosen independently. Now, the desired formula for
$\mathbb{P}_x( N(t)= 1\mid \mathbb{X}_{\tau _1} = y)$ is obtained by
basic calculations involving independent exponentially distributed
random variables, where one has to distinguish the cases
$\deg_{m}(x)\neq\deg_{m}(y)$ and $\deg_{m}(x)=\deg_{m}(y)$.}

The above calculation and Lebesgue's dominated convergence theorem imply
$$\frac{\mathbb{P}_x(N(t)= 1)}{t} \to \frac{1}{m(x)}\sum_{y\in X} b(x,y) = {\rm deg}_m(x),\quad\text{ as }t \searrow 0, $$
showing our claim.
\end{proof}

\begin{lemma}\label{generator} Let $U\subseteq X$ be finite. Then, for all $f\in \ell^2(U,m)$ and $x\in U$, one has
\begin{align*}
    \lim_{t\searrow0}\frac{T_t(v,\theta,U)f(x)-f(x)}{t}=- L^{(U)}_{v,\theta}f(x).
\end{align*}
\end{lemma}
\begin{proof}
We fix an arbitrary $x\in U$ and compute
\begin{align}
\lefteqn{\frac{T_t(v,\theta,U)f(x)-f(x)}{t} } \nonumber \\&=\frac{\mathbb{E}_x\left[1_{\{N(t) = 0\}}\e^{-t v(x)} f(x)\right] -f(x)}{t}
+\frac{\mathbb{E}_x\left[1_{\{N(t) = 1,\mathbb{X}_{\tau_{1}}\in U\}} \e^{\ISS_t(v,\theta|\mathbb{X})} f(\mathbb{X}_t) \right]}{t} + \psi_t(x).\label{eqsemigroup}
\end{align}
The error term $\psi_t(x)$ satisfies $|\psi_t(x)| \leq \varphi_{t,|f|}(x)$ with $\varphi_{t,|f|}$ defined in Lemma~\ref{error}. Therefore, Lemma~\ref{error} implies $\psi_t(x) \to 0$ as $t \searrow 0$. For the first term of the right hand side of \eqref{eqsemigroup}, we have
$$\frac{\mathbb{E}_x\left[1_{\{N(t) = 0\}}\e^{-t v(x)} f(x)\right] -f(x)}{t}=\frac{\e^{-t (v(x) + {\rm deg}_{m}(x))} f(x)  -f(x)}{t}\to- (v(x) + {\rm deg}_{m}(x))f(x)$$
as $t\searrow 0$. Now, let us turn to the second term of the right hand side of \eqref{eqsemigroup}. We obtain
\begin{align*}
\mathbb{E}_x&\left[1_{\{N(t) = 1,\mathbb{X}_{\tau_{1}}\in U\}}\e^{\ISS_t(v,\theta|\mathbb{X})} f(\mathbb{X}_t)\right]\\
&=  \sum_{y \in U} \mathbb{E}_x\left[1_{\{N(t) = 1,\mathbb{X}_{\tau_1}= y\}} \e^{\mathrm{i} \theta(x,y)}\exp\Big(-\tau_1v(x) - (t-\tau_1)v(y)\Big) f(y) \right]\\
&=\sum_{y \in U}\e^{\mathrm{i} \theta(x,y)}f(y) \underbrace{\mathbb{E}_x\left[1_{\{N(t) = 1,\mathbb{X}_{\tau_1}= y\}} \exp\Big(-\tau_1v(x) - (t-\tau_1)v(y)\Big) \right]}_{=:\rho_t(x,y)}.
\end{align*}
Setting
$$
C := 2\max \{|v(x)|\mid x\in U\}
$$
and using $\tau_1 \leq t$ on $\{N(t) = 1\}$, a simple calculation yields
\[
\e^{-tC}\mathbb{P}_x(N(t) = 1,\mathbb{X}_{\tau_1}= y) \leq \rho_t(x,y)\leq \e^{tC}\mathbb{P}_x(N(t) = 1,\mathbb{X}_{\tau_1}= y).
\]
Hence, the same computation as in the proof of Lemma~\ref{error} shows that $\frac{1}{t}\rho_t(x,y)\to b(x,y)/m(x)$ as $t\searrow0$. These two facts and the fact $f \in {C}_c(U)$, as $U$ is finite, imply
$$\frac{1}{t}\mathbb{E}_x\left[ 1_{\{N(t) = 1,\mathbb{X}_{\tau_{1}}\in U\}} \e^{\ISS_t(v,\theta|\mathbb{X})} f(\mathbb{X}_t) \right] \longrightarrow\frac{1}{m(x)}\sum_{y \in U} b(x,y)\e^{\mathrm{i} \theta(x,y)}f(y)\quad\text{as $t\searrow0$},   $$
so, altogether we arrive at
$$
\frac{T_t(v,\theta,U)f(x)-f(x)}{t} \longrightarrow - L^{(U)}_{v,\theta}f(x)\>\>\text{ as $t\searrow0$.}
$$
\end{proof}

With these preparations we can now prove Theorem~\ref{p:finite}.

\begin{proof}[Proof of Theorem~\ref{p:finite}] For finite $U\subseteq X$, we have $\ell^{2}(U,m)=C_{c}(U)$. In particular, $L_{v,\theta}^{(U)}$ is a finite dimensional operator and the convergence
\[
-L_{v,\theta}^{(U)}=\lim_{t\searrow 0}\frac{1}{t}\left(T_t(v,\theta,U)-\mathrm{id}\right)
\]
from Lemma~\ref{generator} holds in the $\ell^{2}(U,m)$ sense. Therefore, the generator of the strongly continuous semigroup  $(T_t(v,\theta,U))_{t\geq 0}$ is given by $L_{v,\theta}^{(U)}$. It follows that $\e^{-tL_{v,\theta}^{(U)}}=T_t(v,\theta,U)$ for all $t\ge0$.
\end{proof}

\subsection{Proof  for closable forms}\label{bew}
\begin{thm}\label{p:FKI_closable}
Let $v$ be a potential  such that $Q^{(c)}_{v,0}$ and $Q^{(c)}_{v,\theta}$ are closable and semi-bounded.
Then for any $f \in \ell^2(X,m)$, $t\geq 0$ and $x\in X$ one has
$$
\e^{-tL_{v,\theta}}f(x) = \mathbb{E}_x\left[1_{\{t<\tau\}}\e^{\ISS_t(v,\theta|\mathbb{X})} f(\mathbb{X}_t)\right].
$$
\end{thm}
\begin{proof}
We prove the asserted formula by using the approximation of $Q_{v,\theta}$ via its restrictions to finite sets. Let $(X_n)_{n\in\IN}$ be an exhausting sequence in the sense of Definition~\ref{exhaus}. Then, Proposition~\ref{approx1} states that
$$\e^{-tL_{v,\theta}}f (x) = \lim_{n \to \infty} \iota_{X_n} \e^{-tL^{(X_n)} _{v,\theta}}\pi_{X_n} f (x).
$$
Combining this with Theorem~\ref{p:finite},  it remains to prove the equation
\begin{gather*}
\lim_{n\to \infty} \mathbb{E}_x\left[1_{\{t<\tau_{X_n}\}} \e^{\ISS_t(v,\theta|\mathbb{X})} \pi_{X_n} f(\mathbb{X}_t)\right] = \mathbb{E}_x\left[1_{\{t<\tau\}} \e^{\ISS_t(v,\theta|\mathbb{X})} f(\mathbb{X}_t)\right]. \label{expconvergence}
\end{gather*}
This will be done in two steps:
\medskip

\emph{Step 1.  $\theta = 0$ and $f \geq 0$}: The sequence $\tau_{X_n}$ converges monotonously increasingly to $\tau$ and  $\pi_{X_n} f(\mathbb{X}_t)$ converges monotonously increasingly to $f(\mathbb{X}_t)$. Hence, the monotone convergence theorem for integrals yields the desired statement.

\emph{Step  2. $\theta$ and $f$ arbitrary}: By the assumption $Q_{v,0}^{(c)}$ gives rise to a self-adjoint semi-bounded operator $L_{v,0}$.
The first step implies
$$\mathbb{E}_x\left[1_{\{t<\tau\}} \e^{\ISS_t(v,0|\mathbb{X})} |f|(\mathbb{X}_t)\right]  = \e^{-tL_{v,0}}|f|(x)< \infty.$$
By Lebesgue's dominated convergence theorem we deduce the desired statement for general $\theta$ and $f$.
\end{proof}

\subsection{Proof of closability of the forms}\label{s:closable_proof}

In this subsection we prove Theorem~\ref{t:closable} from Section~\ref{s:forms} which states that $Q_{v,\theta}^{(c)}$ is closable for all $v\in \mathcal{A}_{0}$.

\begin{proof}[Proof of Theorem~\ref{t:closable}] Let $v\in \mathcal{A}_{0}$ be given. For $n\in \N$ set $v_{n}=v\vee (-n)$ and observe $v_{n}\in \mathcal{B}_{0}$ since $v_{n,-}\in\ell^{\infty}$. By Proposition~\ref{p:closable} the forms $Q_{v_{n},\theta}^{(c)}$ are closable, semi-bounded and their domains satisfy $D(Q_{v_{n},\theta})=D(Q_{v_{+},\theta})$. Moreover, keeping $v\in \mathcal{A}_{0}$ and $C_c(X) \subseteq D(Q_{v_n,\theta})$ in mind, there is some $C>-\infty$ such that $Q_{v_{n},\theta}\ge C$ for all $n$.
Hence, $C\leq Q_{v_{n+1},\theta}\leq Q_{v_{n},\theta}$ for all $n\in\N$.
By monotone convergence of quadratic forms, \cite[Theorem~S.16, p.373]{RS},  we get  $\e^{-tL_{v_{n},\theta}}\to \e^{-tS_{v,\theta}}$, $n\to\infty$, strongly, where  $S_{v,\theta}$ denotes the operator corresponding to the form $s_{v,\theta}$ which is the closure of the largest closable quadratic form that is smaller than the limit form corresponding to $(Q_{v_{n},\theta})_n$.

In order to show closability of $Q_{v,\theta}^{(c)}$, it remains to show that the form domain of $s_{v,\theta}$ includes $C_{c}(X)$ and $s_{v,\theta}$ coincides with $Q_{v,\theta}^{(c)}$ on $C_{c}(X)$.
\medskip

We start by showing that $\mathrm{e}^{-tS_{v,\theta}}$ allows for a Feynman-Kac-It\^o representation:
\medskip

\emph{Claim 1:  For all  $ f \in \ell^2(X,m)$  and $ x \in X$
\begin{align*}
\e^{-t S_{v,\theta}}f (x)= \mathbb{E}_x\left[1_{\{t<\tau\}} \e^{\ISS_t(v,\theta|\mathbb{X})} f(\mathbb{X}_t)\right].
\end{align*}}
By the strong convergence $\e^{-tL_{v_{n},\theta}}\to \e^{-t S_{v,\theta}}$, $n\to\infty$,  it suffices to show that
\begin{align*}
\lim_{n\to \infty} \mathbb{E}_x\left[1_{\{t<\tau\}} \e^{\ISS_t(v_n,\theta|\mathbb{X})} f(\mathbb{X}_t)\right]= \mathbb{E}_x\left[1_{\{t<\tau\}} \e^{\ISS_t(v,\theta|\mathbb{X})} f(\mathbb{X}_t)\right].
\end{align*}
 This, however, can be shown  in
two steps similar to the ones in the proof  of Theorem~\ref{p:FKI_closable} above: We first employ the monotone convergence theorem for $\theta=0$ and $f\ge0$ in the first step and Lebesgue's dominated convergence theorem in the second step.  This proves the claim.\medskip

Next, we compute how the generator of  $\e^{-t S_{v,\theta}}$ acts:

\emph{Claim  2:  For all  $u\in C_c(X)$ and $ x \in \mathrm{supp}\, u$}
\begin{align*}
    \lim_{t\searrow0}\frac{\e^{-t S_{v,\theta}}u(x)-u(x)}{t}=-  \ow L_{v,\theta}u(x).
\end{align*}
Denote $U=\mathrm{supp}\, u$.   Recalling the definitions of   $T_{t}(v,\theta,U)$  and $\varphi_{t,|u|}$  from above  and using  Claim~1  and Lemma~\ref{error}, we obtain
\begin{align*}
    \lim_{t\searrow0}\frac{1}{t} \left|(T_t(v,\theta,U) -\e^{-t S_{v,\theta}}) u(x)\right| \leq2 \lim_{t\searrow0}\varphi_{t,|u|}(x)=0,
\end{align*}
 where the first inequality is readily seen by writing  the semigroups in their
Feynman-Kac-It\^o representation  and splitting up the expectation values into  three parts  corresponding to the events  $\{ N (t) =0\} $, $\{ N (t) = 1\} $
 and $\{ N (t)  \geq 2\} $ as in the proof of Lemma~\ref{generator}.  Then, one immediately sees that the terms for  $\{ N (t) =0\} $  and $\{ N (t) = 1\} $  coincide and the absolute value of each of the terms corresponding to  $\{ N (t)  \geq 2\} $ can be estimated by $\varphi_{t,|u|}(x)$. Having this,  Lemma~\ref{generator} and the observation $L^{(U)}_{v,\theta}u=\ow L_{v,\theta}u$ on $U$ yields the claim.\medskip

To finish the proof, we note that by
Green's formula (Lemma~\ref{Greens formula}), Claim~2 and the semigroup characterization of $s_{v,\theta}$ (\cite[Lemma~1.3.4]{FOT})
\begin{align*}
Q_{v,\theta}^{(c)}(u,u)=\langle u, \ow L_{v,\theta} u\rangle=
    \lim_{t\searrow0}\frac{1}{t}\langle u,u-\e^{-tS_{v,\theta}}u\rangle
    =s_{v,\theta}(u,u),
\end{align*}
where we also used $u\in C_{c}(X)$ in the first two equalities. In particular, this shows that $C_{c}(X)\subseteq D(s_{v,\theta})$. As $Q_{v,\theta}^{(c)}$ is a restriction of a closed form $s_{v,\theta}$, it is closable itself. Semi-boundedness follows as $C_{c}(X)$ is a form core and $v\in \mathcal{A}_{0}$.
\end{proof}

These preparations  readily gives the proof of the main theorem, the Feynman-Kac-It\^o formula for potentials in $\mathcal{A}_{0}$.

\begin{proof}[Proof of Theorem~\ref{main}] By Theorem~\ref{t:closable} the forms $Q_{v,0}^{(c)}$ and $Q_{v,\theta}^{(c)}$ are closable and semi-bounded for $v\in \mathcal{A}_{0}$. Hence, the statement follows by Theorem~\ref{p:FKI_closable}.
\end{proof}


\section{Applications}\label{anwen}

We continue with several applications of the Feynman-Kac-It\^o formula, Theorem~\ref{main}. Remarkably, being equipped with the Feynman-Kac-(It\^{o}) formula, all of the following partially highly nontrivial functional analytic results will be simple consequences of the trivial inequality
\begin{align}
\left|\e^{\ISS_t(v_{1},\theta|\mathbb{X})}\right| \leq\e^{\ISS_t(v_{2},0|\mathbb{X})}\quad\text{ in $\{t < \tau\}$ for all $t\geq 0$},\label{aaa}
\end{align}
and potentials $v_{1}\ge v_{2}$.
This is the main advantage of the path integral formalism.

\subsection{Semigroup formulas}

We will start with the derivation of a probabilistic representation and applications thereof of the integral kernels corresponding to the perturbed magnetic semigroups. To this end, we define the probability measure $\mathbb{P}_{x,y}^t$ on $\{t<\tau\}$ by
\[
  \mathbb{P}_{x,y}^t:=\mathbb{P}_x(\bullet\left| \mathbb{X}_t=y)\right.\>\>\text{for any $x,y\in X$, $t>0$,}
\]
and let $ \mathbb{E}_{x,y}^t$ be the corresponding expected value.
Clearly, the Feynman-Kac-It\^o formula for $v =0$ and $\theta = 0$ implies for $L = L_{0,0}$
\begin{align}
\mathbb{P}_{x}(A)&=\sum_{y\in X}\mathbb{P}_{x,y}^t(A)\mathbb{P}_x(\mathbb{X}_t=y)
=\sum_{y\in X}\mathbb{P}_{x,y}^t(A)\e^{-t L}(x,y)m(y)\label{condi}
\end{align}
for any event $A\subset \{t<\tau \}$. Therefore, we obtain
$$
{\Lp}^1\left(\{t<\tau\},\mathbb{P}_{x}\right)\subset {\Lp}^1\left(\{t<\tau\},\mathbb{P}_{x,y}^t\right).
$$

\begin{thm}\label{int2} Let $v\in\mathcal{A}_{0}$. Then for all $t >0$ and $x,y\in X$ one has
\begin{align*}
\e^{-tL_{v,\theta}}(x,y) =\f{1}{m(y)}\mathbb{P}_x(\mathbb{X}_t=y) \mathbb{E}_{x,y}^t\left[\e^{\ISS_t(v,\theta|\mathbb{X})} \right]= \e^{-t L}(x,y)\mathbb{E}^t_{x,y}\left[\e^{\ISS_t(v,\theta|\mathbb{X})} \right],
\end{align*}
in particular,
\begin{align*}
\mathrm{tr}\left[\e^{-tL_{v,\theta}}\right] &= \sum_{x\in X}\mathbb{P}_x(\mathbb{X}_t=x) \mathbb{E}^t_{x,x}\left[\e^{\ISS_t(v,\theta|\mathbb{X})} \right]=\sum_{x\in X}\e^{-t L}(x,x) \mathbb{E}^t_{x,x}\left[\e^{\ISS_t(v,\theta|\mathbb{X})} \right]m(x)\in[0,\infty].
\end{align*}
\end{thm}
\begin{proof}
The Feynman-Kac-It\^{o} formula in combination with (\ref{condi}) directly implies the first formula.
It only remains to prove the formula for the trace. Clearly, by the semigroup property and self-adjointness, $\mathrm{tr}\left[\e^{-tL_{v,\theta}}\right]$ is equal to the Hilbert-Schmidt norm of $\e^{-\f{t}{2}L_{v,\theta}}\e^{-\f{t}{2}L_{v,\theta}}$, which in view of the formula for $\e^{-tL_{v,\theta}}(x,y)$ and the semigroup property and symmetry of the latter precisely has the asserted form.
\end{proof}

\subsection{Kato\rq{}s inequality}

The following theorem includes a general version of Kato\rq{}s inequality and applications thereof. We refer the reader to \cite{bg} for probabilistic aspects of Kato\rq{}s inequality on noncompact Riemannian manifolds, and to \cite{DM} for a direct proof of Kato's inequality on graphs (in a more restrictive setting though). Moreover, some of the results below are also contained in \cite{Gol} for locally finite graphs.

\begin{thm}\label{aw}\emph{(Kato\rq{}s inequality)} Let $v_1,v_{2}\in\mathcal{A}_{0}$ be potentials such that $v_1\geq v_2$. Then the following assertions hold:
\begin{itemize}
\item[(a)] For all $t\geq 0$, $f \in \ell^2(X,m)$ and $x\in X$, one has $$\left|\e^{-tL_{v_1,\theta}}f(x) \right|\leq  \e^{-tL_{v_2,0}}|f|(x).$$
     In particular, for all $x,y\in X$ and $t>0$, one has
$$
|\e^{-tL_{v_1,\theta}}(x,y)|\leq  \e^{-tL_{v_2,0}}(x,y),\>\> \mathrm{tr}\left[\mathrm{e}^{-tL_{v_1,\theta}}\right] \leq \mathrm{tr}\left[\mathrm{e}^{-tL_{v_2,0}}\right].
$$

\item[(b)] For any $h\in{\D}(Q_{v_1,\theta})$, it holds that $\left|h\right|\in {\D}(Q_{v_2,0})$ and
$Q_{v_1,\theta}(h)\geq Q_{v_2,0}(|h|)$.

\item[(c)] One has
$\inf\sigma(L_{v_1,\theta})\geq \inf\sigma(L_{v_2,0})$.

\item[(d)] For any $f \in \ell^2(X,m)$, $\lambda\in\IC$ with $\mathrm{Re}(\lambda)>\min\sigma(L_{v_1,\theta})$, $x\in X$,
$$\left|(L_{v_1,\theta}+\lambda)^{-1}f(x)\right|\leq (L_{v_2,0}+\lambda)^{-1}|f|(x).$$

\item[(e)] If $L_{v_2,0}$ has a compact resolvent, then $L_{v_1,\theta}$ has a compact resolvent.
\end{itemize}
\end{thm}

\begin{proof} Assertion (a) is implied by Theorem~\ref{int2} together with (\ref{aaa}). Statement (b) follows from (a) and the semigroup characterizations of $Q_{v_2,0}$ and $Q_{v_1,\theta}$, see \cite[Lemma~1.3.4]{FOT}, and (c) follows from (b) and the variational characterization of the bottom of the spectrum, see \cite{RS}, (or simply cf. \cite[Theorem D.6]{bg} for both (b) and (c)).
Statement (d) is a direct consequence of (a)  and the Laplace\rq{}s formula for the resolvents.
For (e) notice that the operators $\e^{-tL_{v_2,0}}$ are positivity improving for all $t > 0$ by the Feynman-Kac formula and $(L_{v_2,0}+\lambda)^{-1}$ are positivity improving for  all $\lambda > \inf\sigma(L_{v_2,0})$ by the Laplace formula for resolvents. Thus, the statement of the theorem follows from (d)  by using Pitt\rq{}s theorem (cf. Theorem~\ref{pits}).
\end{proof}

\subsection{Golden-Thompson inequality}
The following is a discrete analogue of the Golden-Thompson inequality.

\begin{thm}\label{gold} \emph{(Golden-Thompson inequality)} Let $v_1,v_{2}\in\mathcal{A}_{0}$ be potentials such that $v_1\geq v_2$. Then for any $t>0$ one has
\begin{align*}
\mathrm{tr}\left[\e^{-tL_{v_1,\theta}}\right]&\leq \sum_{x\in X}e^{-tL}(x,x)
 \e^{-tv_2(x)}m(x)\leq C(t)\sum_{x\in X}\e^{-tv_2(x)}\in[0,\infty],
\end{align*}
where
$$C(t):=\sup_{x\in X}\e^{-t L}(x,x)m(x)\leq1.$$
\end{thm}

For the proof of the Golden-Thompson inequality, Theorem~\ref{gold}, we need the following monotonicity property of the trace,  which should also be of an independent interest as well.

\begin{prop}\label{mon}  Let $v\in\mathcal{A}_{0}$. Then for any exhausting sequence $(X_n)_{n\in \IN}$ one has
$$
\mathrm{tr}\left[\e^{-tL^{(X_n)}_{v,0}}\right]\nearrow \mathrm{tr}\left[\e^{-tL_{v,0}}\right]\>\> \text{as $n\to \infty$ for all $t>0$.}
$$
\end{prop}

\begin{proof} Combining Theorem~\ref{p:finite} with (\ref{condi}) easily implies
$$
\mathrm{tr}\left[\e^{-tL^{(X_n)}_{v,0}}\right]=\sum_{x\in X_n}\e^{-t L}(x,x) \mathbb{E}^t_{x,x}\left[1_{\{ t<\tau_{X_n}\} }\e^{-\int^t_0v(\mathbb{X})\Id s } \right]m(x),
$$
which, using Theorem~\ref{int2}, tends to $\mathrm{tr}\left[\e^{-tL_{v,0}}\right]$ in view of monotone convergence.
\end{proof}

\begin{proof}[Proof of Theorem~\ref{gold}] In view of Theorem~\ref{aw} (a), we have $\mathrm{tr}\left[\e^{-tL_{v_1,\theta}}\right]\leq \mathrm{tr}\left[\e^{-tL_{v_2,0}}\right]$. Let $(X_n)_{n\in\IN}$ be an exhausting sequence. Then applying the operator-version of Golden-Thompson inequality (Theorem~\ref{gti}) to $q\rq{}=Q^{(X_n)}_{0,0}$, $q\rq{}\rq{}=q_{v_2}$ in the Hilbert space $\ell^2(X_n,m)$, where $Q^{(X_n)}_{0,v_2}=Q^{(X_n)}_{0,0}+q_{v_2}$ is trivial in view of the finiteness of $X_n$, we get the inequality in
\begin{align}
\mathrm{tr}\left[\mathrm{e}^{-t L^{(X_n)}_{v_2,0}}\right]&\leq \mathrm{tr}\left[\mathrm{e}^{-\f{t}{2} L^{(X_n)}_{0,0}} \mathrm{e}^{-tv_2}\mathrm{e}^{-\f{t}{2} L^{(X_n)}_{0,0}}\right]\nn\\
&= \mathrm{tr}\left[\left(\mathrm{e}^{-\f{t}{2}v_2} \mathrm{e}^{-\f{t}{2} L^{(X_n)}_{0,0}}\right)^* \left(\mathrm{e}^{-\f{t}{2}v_2}\mathrm{e}^{-\f{t}{2} L^{(X_n)}_{0,0}}\right)\right]\nn\\
&=\sum_{x\in X_n} \mathrm{e}^{-tv_2(x)} \sum_{y\in X_n}\mathrm{e}^{-\f{t}{2} L^{(X_n)}_{0,0}}(x,y)\mathrm{e}^{-\f{t}{2} L^{(X_n)}_{0,0}}(y,x) m(y) m(x)\nn\\
&=\sum_{x\in X_n} \mathrm{e}^{-t L^{(X_n)}_{0,0}}(x,x)\mathrm{e}^{-tv_2(x)}  m(x)\>\>\text{ for all $n$.}\label{hig}
\end{align}
Here we have used self-adjointness and semigroup properties, as well as
$$
\left(\mathrm{e}^{-\f{t}{2}v_2}\mathrm{e}^{-\f{t}{2} L^{(X_n)}_{0,0}}\right)(x,y)=\mathrm{e}^{-\f{t}{2}v_2(x)}\mathrm{e}^{-\f{t}{2} L^{(X_n)}_{0,0}}(x,y)\>\>\text{ for all $(x,y)\in X_n\times X_n$}.
$$
Noting that
$$
1_{X_n\times X_n}(x,x)\mathrm{e}^{-t L^{(X_n)}_{0,0}}(x,x)\nearrow \mathrm{e}^{-tL}(x,x)\>\>\text{ for all $x\in X$ as $n\to\infty$,}
$$
monotone convergence implies that the right-hand side of (\ref{hig}) tends to the term in the middle of the asserted inequality as $n\to\infty$.  In view of Proposition~\ref{mon}, this  completes the proof of the first inequality. For the second inequality we note that $\mathrm{e}^{-tL}(x,x)m(x)=\mathbb{P}_x(\mathbb{X}_t=x)\leq 1$.
\end{proof}

\begin{remark} We refer the reader to \cite[Theorem~9.2]{Si05} for an $\IR^m$-version of the Golden-Thompson inequality, which uses a very different proof. Note that in this particular case, the Golden-Thompson inequality can be rewritten as a phase space bound. This has the important physical consequence that the quantum mechanical partition function is always bounded from above by the corresponding classical partition function.
\end{remark}

\subsection{The form domain}

Finally, we use the Feynman-Kac-It\^{o} formula to derive an explicit description of the form domain of $Q_{v,\theta}$ under suitable assumptions on the potential.

\begin{thm}\label{t:formdomain_theta} For any $v\in \mathcal{B}_{0}$, one has
$Q_{v,\theta}=Q_{0,\theta}+q_{v}$, in particular, $D(Q_{v,\theta})=D(Q_{0,\theta})\cap\ell^{2}(X,|v|m)$.
\end{thm}

\begin{proof}
By Proposition \ref{p:closable} it suffices to show the statement for $v\ge0$. We prove a Feynman-Kac-It\^{o} formula for $Q_{0,\theta}+q_{v}$  in order to conclude the assertion using Theorem~\ref{main}. To this end, denote the operator arising from the form sum $Q_{0,\theta}+q_{v}$ by $L_{0,\theta}+v$.\\
With $v_n:=v\wedge n\in \ell^{\infty}(X)$ we have $
0\leq v_n\nearrow v$  as $n\to\infty$
and it follows from monotone convergence for integrals that $Q_{v_n,\theta} =Q_{0,\theta}+q_{v_n} \nearrow Q_{0,\theta}+q_v$ as $n\to\infty$ in the sense of monotone convergence of quadratic forms. By \cite[Theorem~S.14, p.373]{RS} we have that $Q_{0,\theta}+q_{v}$ is closed and
\begin{align*}
\lim_{n\to\infty}\mathrm{e}^{-t(L_{0,\theta}+v_{n})}f(x) =\mathrm{e}^{-t(L_{0,\theta}+v)}f(x)
\end{align*}
for all $f\in \ell^{2}(X,m)$ and $x\in X$.
Thus, in view of $Q_{v_{n},\theta}=Q_{0,\theta}+q_{v_{n}}$ and $L_{v_{n},\theta}=L_{0,\theta}+v_{n}$ (as $v_{n}$ is bounded) it only remains to prove
\begin{align*}
\lim_{n\to\infty} \mathbb{E}_x\left[1_{\{t<\tau\}} \e^{\ISS_t(v_{n},\theta|\mathbb{X})} f(\mathbb{X}_t)\right] =\mathbb{E}_x\left[1_{\{t<\tau\}} \e^{\ISS_t(v,\theta|\mathbb{X})} f(\mathbb{X}_t)\right].
\end{align*}
which,  however, follows by Lebesgue's dominated convergence.
\end{proof}

We finish with a corollary of the theorem above. For $v$ bounded below, recall the form $Q^{\max}_{v,\theta}:\ell^{2}(X,m)\to(-\infty,\infty]$  in the proof of Proposition~\ref{p:closable}, which is given by
\begin{align*}
    Q^{\max}_{v,\theta}(f)=\frac{1}{2}\sum_{x,y \in X} b(x,y) |f(x)-\e^{\mathrm{i} \theta(x,y) }f(y)|^{2} + \sum_{x\in X} v(x)|f(x)|^{2} m(x).
\end{align*}
It is bounded below and closed.

\begin{coro}If $Q_{0,\theta}=Q^{\max}_{0,\theta}$, then one has $Q_{v,\theta}=Q^{\max}_{v,\theta}$ for all $v$ bounded below.
\end{coro}
\begin{proof} As, obviously, $Q^{\max}_{v,\theta}=Q_{0,\theta}^{\max}+q_{v}$, the statement follows by the theorem above.
\end{proof}



\appendix
\section{Pitt's theorem}

\begin{thm}\label{pits} Let $p_1\in (1,\infty)$, $p_2\in [1,\infty]$ and let
$A,B:{\Lp}^{p_1}(M,\mu)\to{\Lp}^{p_2}(M,\mu)$ be bounded operators
such that $A$ is positivity preserving and such that one has $|Bf|\leq A|f|$ for any $f\in{\Lp}^{p_1}(M,\mu)$. Then $B$ is a compact operator, if $A$ is a compact operator.
\end{thm}

This highly nontrivial fact on operator domination goes back to L.D. Pitt \cite{pitt}.

\section{An abstract Golden-Thompson inequality}

\begin{thm}\label{gti} Let $q\rq{}$, $q\rq{}\rq{}$ be densely defined, closed, symmetric and semi-bounded sesquilinear forms on a common Hilbert space. Assume that $q:=q\rq{}+q\rq{}\rq{}$ is densely defined and denote the semigroups corresponding to $q\rq{},q\rq{}\rq{}$ and $q$ by $(T_t\rq{})_{t\geq 0}$, $(T_t\rq{}\rq{})_{t\geq 0}$ and $(T_t)_{t\geq 0}$, respectively. Then one has
\[
\mathrm{tr} [T_t]\leq \mathrm{tr}\left[T_{t/2}\rq{}  \ T_{t}\rq{}\rq{} \  T_{t/2}\rq{}\right]\>\>\text{ for all $t\geq 0$.}
\]
\end{thm}

This result follows from Corollary 3.9 in \cite{fumio}. 
Note that the above fact is even nontrivial for finite dimensional operators.

\section{Mosco-convergence}

Let $(H_k,\as{\cdot,\cdot}_k)$, $k \in\IN$, and $(H, \as{\cdot,\cdot})$ be Hilbert spaces with corresponding norms $\|\cdot\|_k$ and $\|\cdot\|$ respectively. Suppose $(q_k,D(q_k))$ and $(q,D(q))$ are densely defined closed symmetric sesquilinear forms on $H_k$ and $H$, respectively, which are bounded below by a constant $C> -\infty$ which is \emph{uniform} in $k$. Each $q_k$ is understood to be defined on the whole space $H_k$ by the convention $q_k(u) = \infty$ whenever $u \in H_k \setminus D(q_k)$. Furthermore, we suppose that there exist bounded operators $\iota_k:H_k \to H$ such that $\pi_k := \iota_k ^*$ is a left inverse of $\iota_k$, that is
$$\as{\pi_kf, f_k}_k = \as{f,\iota_kf_k}\text{ and } \pi_k \iota_kf_k = f_k, \text{ for all} f\in H, f_k\in H_k.$$
Moreover, we assume that $\pi_k$ satisfies
$$\sup_{k\in\IN}\|\pi_k\|< \infty \text{ and } \lim_{k\to \infty} \|\pi_kf\|_k = \|f\|.$$

\begin{definition} \label{mosco}
In the above situation, we say that $q_k$ is \emph{Mosco convergent} to $q$ as $k\to\infty$ in the generalized sense, if the following conditions hold:
\begin{itemize}
\item[(a)] If $u_k \in H_k$, $u \in H$ and $\iota_ku_k \to u$ weakly in $H$, then
$$\liminf_{k \to \infty}\left(q_k(u_k) + C\|u_k\|_k^2\right) \geq q(u) + C\|u\|^2.$$
\item[(b)] For every $u \in H$ there exist $u_k \in H_k$, such that $\iota_k u_k \to u$ in $H$ and
$$\limsup_{k \to \infty}\left(q_k(u_k) + C \|u_k\|_k^2\right) \leq q(u) + C \|u\|^2.$$
\end{itemize}
\end{definition}

Let $(T^{(k)}_t)_{t \geq 0}$ denote the semigroup associated with $q_k$ and let $(T_t)_{t\geq 0}$ be the semigroup of $q$.  For positive forms the following theorem which characterizes Mosco convergence can be found in the appendix of \cite{CKK}. However, this result immediately extends to the situation of forms with uniform lower bound.

\begin{thm} \label{mosco.char}
In the above situation, the following assertions are equivalent:
\begin{itemize}
\item[(a)] $q_k$ is Mosco convergent to $q$ as $k\to\infty$ in the generalized sense.
\item[(b)] One has $\iota_k T^{(k)}_t\pi_k \to T_t$ as $t\to\infty$ strongly and uniformly on any finite time interval.
\end{itemize}
\end{thm}
\begin{proof}
Consider the positive quadratic forms $\tilde{q}_k = q_k + C\|\cdot\|^2$ and $\tilde{q} = q + C\|\cdot\|^2$. Obviously their semigroups $\ow{T}_t^{(k)}$ and $\ow{T}_t$ satisfy
$$\ow{T}_t^{(k)} = \mathrm{e}^{-tC}T_t^{(k)} \text { and } \ow{T}_t  = \mathrm{e}^{-tC}T_t.$$
Combining this and the characterization of Mosco convergence for positive forms (Theorem~8.3 of \cite{CKK}) we can deduce the result.
\end{proof}

\textbf{Acknowledgement}. The authors want to express their gratitude to Daniel Lenz for very inspiring and important hints. We are also grateful to Ognjen Milatovic for a very fruitful correspondence and to Hendrik Vogt, Sebastian Haeseler for most helpful discussions. BG has been financially supported by the Sonderforschungsbereich 647: Raum-Zeit-Materie. Moreover, MK acknowledges the financial support of the German Science Foundation (DFG), Golda Meir Fellowship, the Israel Science Foundation (grant no. 1105/10 and  no. 225/10) and BSF grant no. 2010214. Furthermore, MS acknowledges the financial support of the European Science Foundation (ESF) within the project ``Random Geometry of Large Interacting Systems and Statistical Physics''.

\end{document}